\pdfoutput=1
\documentclass[letterpaper]{article} 
\usepackage{times}  
\usepackage{helvet}  
\usepackage{courier}  
\usepackage[hyphens]{url}  
\usepackage[dvipdfmx]{graphicx} 
\usepackage[margin=1.5in]{geometry}
\usepackage[dvipsnames]{xcolor}
\usepackage{tikz}
\usepackage{tikzscale}
\usepackage{enumitem}
\usepackage{natbib}
\usetikzlibrary{arrows.meta}
\usetikzlibrary{decorations.pathreplacing}
\usepackage{natbib}  
\usepackage{caption} 
\usepackage{blindtext}
\usepackage{algorithm}
\usepackage{algorithmic}

\usepackage{microtype}
\usepackage{subfigure}
\usepackage{booktabs} 

\usepackage{amsmath}
\usepackage{amssymb}
\usepackage{mathtools}
\usepackage{amsthm}

\usepackage[capitalize,noabbrev]{cleveref}

\theoremstyle{plain}
\newtheorem{theorem}{Theorem}

\newtheorem{lemma}[theorem]{Lemma}

\theoremstyle{definition}

\theoremstyle{remark}

\usepackage{enumitem}
\usepackage{natbib}  
\usepackage{caption} 
\usepackage{blindtext}
\usepackage{algorithm}
\usepackage{algorithmic}

\usepackage{microtype}

\usepackage{subfigure}
\usepackage{booktabs} 

\usepackage[capitalize,noabbrev]{cleveref}

\newcommand{\alglong}{\texttt{ConfThreshGreedy}\xspace}
\newcommand{\alglongmono}{\texttt{ConfThreshGreedy2}\xspace}
\newcommand{\algmono}{\texttt{CTG2}\xspace}
\newcommand{\alg}{\texttt{CTG}\xspace}
\newcommand{\prob}{MSMC\xspace}

\newcommand{\conf}{C_t}
\newcommand{\singla}{\texttt{ExpGreedy}\xspace}

\newcommand{\mE}{\mathbb{E}}

\newcommand{\samp}{\texttt{CS}\xspace}
\newcommand{\samplong}{\texttt{Confident Sample}\xspace}
\newcommand{\thresholdlong}{\texttt{Threshold Greedy}\xspace}
\newcommand{\threshold}{\texttt{TG}\xspace}

\newcommand{\cdg}{\texttt{CDG}\xspace}
\newcommand{\contialglong}
{\texttt{ConfContinuousThreshGreedy}\xspace}
\newcommand{\contialg}
{\texttt{CCTG}\xspace}
\newcommand{\contisublong}
{\texttt{Decreasing-Threshold Procedure}\xspace}
\newcommand{\contisub}
{\texttt{DTP}\xspace}

\newcommand{\vect}{\textbf}
\newcommand{\appendixtitle}[1]{
	\begin{center}
		\LARGE \bf #1
	\end{center}
}
\newcommand{\sampnew}{\texttt{CS}\xspace}
\newcommand{\sampnewlong}{\texttt{Confident Sample}\xspace}
\usepackage{multicol}
\usepackage[textsize=tiny]{todonotes}

\usepackage{authblk}
\title{A Threshold Greedy Algorithm for Noisy Submodular Maximization}
\author[1]{Wenjing Chen}
\author[1]{Shuo Xing}
\author[1]{Victoria G. Crawford}
\affil[1]{Department of Computer Science \& Engineering, Texas A\&M University}
\date{}                     
\begin{document}

\maketitle

\begin{abstract}
We consider the maximization of a submodular objective function $f:2^U\to\mathbb{R}_{\geq 0}$, where the objective $f$ is not accessed as a value oracle but instead subject to noisy queries. We introduce a versatile adaptive sampling procedure called \samplong (\samp), which determines whether the marginal gain of the function $f$ is approximately above or below an input threshold with high probability in as few noisy samples as possible. Using \samp as a subroutine, we propose sample efficient algorithms for monotone submodular maximization with cardinality and matroid constraints, as well as unconstrained non-monotone submodular maximization. The proposed algorithms achieve approximation guarantees arbitrarily close to those of the standard value oracle setting. We further provide an experimental evaluation on real instances of \prob and demonstrate the sample efficiency of our proposed algorithm relative to alternative approaches.
\end{abstract}

\section{Introduction}


Submodularity is a property of set functions that arises in many applications such as cut functions in graphs \cite{balkanski2018non}, coverage functions \cite{bateni2017almost}, data summarization objectives \cite{tschiatschek2014learning}, information theoretic quantities such as mutual information \cite{iyer2021generalized}, and viral marketing in social networks \cite{kempe2003maximizing}. A function $f:2^U\to\mathbb{R}_{\geq 0}$ defined over subsets of the universe $U$ of size $n$ is submodular if for all $X\subseteq Y\subseteq U$ and $u\notin Y$, $f(Y\cup\{u\}) - f(Y)\leq f(X\cup \{u\})-f(X)$. In addition, in many applications of submodular functions $f$ is monotone \citep{tschiatschek2014learning,iyer2021generalized,kempe2003maximizing}, meaning that for all $X\subseteq Y\subseteq U$, $f(X)\leq f(Y)$.

Approximation algorithms for submodular optimization problems have received a wealth of attention \cite{nemhauser1978analysis,badanidiyuru2014fast,mirzasoleiman2015lazier,balkanski2019exponential,buchbinder2015tight,calinescu2011maximizing}. Proposed algorithms typically are assumed to have value oracle access to $f$. That is, $f$ is a black box that can be queried for any $X\subseteq U$, and the value of $f(X)$ is returned. 

However, in many optimization scenarios, we can only make noisy queries from some random distribution to estimate the objective. For example, in applications such as data summarization with human feedback \cite{singla2016noisy}, influence maximization \cite{kempe2003maximizing}, feature selection tasks \cite{krause2005near}, querying the exact value of $f$ is unrealistic, and instead a more realistic assumption is that we can query $f$ subject to some noise  \cite{horel2016maximization,singla2016noisy,hassidim2017submodular,qian2017subset,crawford2019submodular,huang2022efficient}. And in a related setting, submodular optimization algorithms that leverage the multilinear extension $F$ of the function $f$ may only be able to access $F$ via noisy random samples \cite{calinescu2011maximizing,badanidiyuru2014fast}. In this setting, the common approach is to use existing submodular optimization algorithms and apply the fact that the objective can be evaluated to arbitrary precision by taking sufficiently many samples and applying concentration inequalities in order to achieve a fixed-precision\footnote{See Section \ref{sec:prelim} for a definition of fixed-precision approximation. } approximation of the objective function \cite{calinescu2011maximizing}. However, modern massive datasets demand algorithms that are as efficient as possible in terms of runtime, and in the case of submodular optimization algorithms, the main computation time bottleneck for the above approach would be the noisy queries to $f$.

Motivated by the above, our main insight is that an algorithm doesn't necessarily need to approximate $f$ with such fine precision at every query in order to find a solution with an approximation guarantee comparable to the exact value oracle setting. Instead, we propose methods of adaptively approximating the function $f$ based on decisions that the algorithm must make, with an emphasis on minimizing the total number of noisy queries. In particular, the contributions of the paper are as follows:
\begin{enumerate}[noitemsep]
\item We propose the adaptive sampling algorithm \samplong (\samp) in Section \ref{sec:sampling}, which can be used to determine if the mean of a random variable $X$ is approximately above or below a given threshold $w$ with high probability, in relatively few random samples.
Intuitively, the required number of samples is inversely proportional to the distance between $\mE X$ and $w$, and therefore we can significantly decrease the number of samples relative to the fixed-precision approach by sampling less when the distance is large. This sampling algorithm serves as a versatile tool that is used in many noisy submodular algorithms in the subsequent sections.

\item Using \samp as a subroutine, we propose algorithms for various submodular maximization problems under the noisy setting. The proposed algorithms exhibit an improved sample complexity compared with fixed-precision approximation. The results are listed as follows: 
\begin{enumerate}[noitemsep]
   \item First of all, we address the problem of Monotone Submodular Maximization with Cardinality constraint (MSMC) in Section \ref{sec:monotone}. Given a cardinality constraint (or ``budget'') $\kappa$, and noisy access to a monotone submodular function $f:2^U\to\mathbb{R}_{\geq 0}$, \prob is to find the set $\arg\max\{f(X): X\subseteq U, |X|\leq \kappa\}$. 
   We propose the \alglong algorithm (\alg) for MSMC, which achieves an approximation guarantee arbitrarily close to $1-1/e$ with high probability. 
 The detailed results on sample complexity are presented in Theorem \ref{mainthm} of Section \ref{sec:monotone}. 
    \item In Section \ref{sec:nonmono}, we propose \texttt{Confident Double Greedy} (\texttt{CDG}) for Unconstrained Submodular Maximization (USM). Given noisy access to a non-negative submodular function $f:2^U\rightarrow \mathbb{R}^+$ which is \textit{not necessarily monotone}, the goal is to find a subset $S\subseteq U$ that maximizes $f(S)$. \texttt{CDG} follows a similar idea to the deterministic algorithm in \cite{buchbinder2015tight}, and is proved to satisfy an approximation guarantee arbitrarily close to $1/3$ with high probability. The theoretical guarantee on sample complexity is presented in Theorem \ref{thm:nonmono}.
    \item In Section \ref{sec:matroid}, the algorithm \contialglong (\contialg) is proposed for the problem of Monotone Submodular Maximization with Matroid constraint (MSMM). MSMM is to find the solution of $\arg\max_{S\subseteq\mathcal{M}}f(S)$, where $\mathcal{M}$ is a matroid defined on subsets of the ground set $U$, and that $f$ is monotone and submodular. \contialg accesses the multilinear extension of $f$ via noisy samples, since the multilinear extension can be difficult to compute in general \cite{calinescu2011maximizing,badanidiyuru2014fast}. We prove that the proposed algorithm achieves an approximation ratio of arbitrarily close to $1-1/e$ with high probability. In addition, we demonstrate that \contialg has an improved sample complexity compared with the one proposed in \cite{badanidiyuru2014fast}. 
\end{enumerate}

    \item Finally, as a demonstration of our approach, we experimentally analyze \alg on instances of noisy data summarization and influence maximization. We compare \alg to several alternative methods including the algorithm of \citet{singla2016noisy} which is discussed in more detail in Section \ref{sec:relatedwork} and in the appendix.
    \alg is demonstrated to be a practical choice that can save many samples relative to alternative approaches.
\end{enumerate}


\subsection{Related Work}
\label{sec:relatedwork}
Approximation algorithms for the maximization of a submodular objective function subject to various constraints have been extensively studied in the literature \cite{nemhauser1978analysis,badanidiyuru2014fast,mirzasoleiman2015lazier,calinescu2011maximizing} with the assumption of oracle access to $f$. The runtime of these algorithms is generally measured in queries to $f$ as this is the main bottleneck (see Section \ref{appdx:related_work} for a more comprehensive discussion on the runtime of algorithms for various submodular optimization problems).

While there are many works assuming value oracle access to $f$, algorithms developed assuming noisy access to $f$ are relatively less explored \cite{horel2016maximization,singla2016noisy,hassidim2017submodular,qian2017subset,crawford2019submodular,huang2022efficient}. One related setting to ours is that we have noisy access to $f$, but this noise is \textit{persistent} \cite{horel2016maximization,hassidim2017submodular,qian2017subset,crawford2019submodular,huang2022efficient}, i.e. repeated samples cannot be taken to diminish the noise. Our proposed approach does not apply to the persistent noise setting, since repeated sampling is a central part of our algorithm. Another related but different setting is that of stochastic submodular optimization \cite{karimi2017stochastic,staib2019distributionally,ozcan2023stochastic} which assumes the optimization objective $f$ is the expectation over some unknown distribution over a set of monotone submodular functions. Therefore a sample average function can be built, which is also monotone and submodular, and algorithms run on it. In contrast, in our setting, it is only assumed that we can sample noisy queries at each subset $X\subseteq U$.

The algorithm \singla of \citet{singla2016noisy} is for a noisy setting identical to ours and is developed for the MSMC problem specifically. \singla also incorporates an adaptive sampling approach. In particular, their algorithm combines the standard greedy algorithm with the best arm identification problem found in combinatorial bandit literature \cite{chen2014combinatorial}. Their approach is still very different from ours, and an extensive comparison of our algorithms and results with \cite{singla2016noisy} are presented in the appendix, as well as an experimental comparison in Section \ref{sec:exp}. 

The intuition behind \samp is similar to the best-arm-identification problem in the multi-armed bandit literature \cite{kalyanakrishnan2012pac}.
Both the algorithm LUCB of \cite{kalyanakrishnan2012pac} and \samp share a common underlying intuition: they leverage the difference between expectations to reduce the number of noisy queries required. In LUCB, this difference is between the expectation of the optimal arm and other arms, while in \samp, it is between the expectation of the input variable and the threshold value $w$.
\section{Preliminary Definitions and Notations}
\label{sec:prelim}
In this section, we lay the groundwork definitions and notations for the remainder of the paper.
Throughout this paper, we assume $f:2^U\to\mathbb{R}_{\geq 0}$ is submodular. $U$ is the ground set of size $n$.
Let us denote the marginal gain of adding element $u\in U$ to a set $X\subseteq U$ as $\Delta f(X,u)$, i.e., $\Delta f(X,u):=f(X\cup\{u\})- f(X)$.

We first define the noisy model of access to $f$. In particular, given any subset $X\subseteq U$ and $u\in U$, independent samples can be taken from the distribution $\mathcal{D}(X,u)$ to obtain noisy evaluations of $\Delta f(X,u)$. In this paper, we denote the random variable following the distribution of $\mathcal{D}(X,u)$ as $\widetilde{\Delta f}(X,u)$. We assume the following properties about the distribution $\mathcal{D}(X,u)$: (i) $\mathbb{E} [\widetilde{\Delta f}(X,u)]=\Delta f(X,u)$; and (ii) $\widetilde{\Delta f}(X,u)$ are bounded in the range of $[0, R]$ for all $X,u$ (or in some results, they are assumed to be $R$-sub-Gaussian).\footnote{ A random variable that is bounded within the interval $[0,R]$ can be demonstrated to be $R/2$ sub-Gaussian. Consequently, the assumption of a random variable being sub-Gaussian is more general than that of boundedness.} In addition, in applications where instead we have noisy queries directly to $f$ instead of the marginal gain, this also satisfies our setting (see Section \ref{appdx:related_work} in the appendix of the supplementary material for more details).


\textbf{Fixed $\epsilon$-approximation. }Given any random variable $X$, an estimate $\widehat{X}$ is a $\textit{fixed $\epsilon$-approximation}$ of $X$ if $\mE X-\epsilon\leq \widehat{X} \leq \mE X+\epsilon$. Notice that for any $X$ that is $R$-sub-Gaussian, we can take $O\left(\frac{R^2}{\epsilon^2}\log \frac{1}{\delta}\right)$ samples and the sample average is a fixed $\epsilon$-approximation of $X$ with probability at least $1-\delta$ by an application of Hoeffding's Inequality (Lemma \ref{hoeffding} in the appendix in the supplementary material).

\textbf{Multi-linear extension}. For any submodular objective $f$, the multi-linear extension of $f$ is defined as $\vect{F}$, i.e., $\vect{F}(\vect{x})=\sum_{S\subseteq U}\prod_{i\in S}x_i\prod_{j\notin S}(1-x_j)f(S)$. Here we define $S(\vect{x})$ to be a random set that contains each element $i\in U$ with probability $x_i$, then by definition, we have that $\vect{F}(\vect{x})=\mE f(S(\vect{x}))$.


\section{Confident Sampling Algorithm}
\label{sec:sampling}

In this section, we propose and analyze the \samplong (\samp) algorithm. \samp is used in order to determine if the expected value of a random variable $X$ is approximately above or below a threshold value with high probability. \samp works for any random variable that is $R$-sub-Gaussian (see Theorem \ref{thm:sampling}) or bounded in the range of $[0,R]$ (see Theorem \ref{thm:sampling2}). 
In Sections \ref{sec:monotone}, \ref{sec:nonmono}, and \ref{sec:matroid}, we show that \samp is useful as a subroutine for a variety of submodular maximization algorithms where we only have noisy access to the marginal gains.

We now describe \samp. \samp takes as input failure probability $\delta\in\mathbb{R}_{>0}$, threshold error parameter $\epsilon\in\mathbb{R}_{>0}$, a threshold value $w\in\mathbb{R}_{>0}$, the unknown distribution $\mathcal{D}_X$ of the random variable $X$ which it accesses via independent random samples, and the sub-Gaussian parameter $R$. \samp iteratively takes at most $N_1$ samples from $\mathcal{D}_X$, while maintaining a sample average and a confidence interval. In particular $\widehat{X}_t$ is the sample average after taking $t$-th samples of $X$, i.e., $\widehat{X}_t=\frac{1}{t}\sum_{i=1}^tX_i$ where $X_i$ is the $i$-th random sample of $X$. The confidence region, after taking the $t$-th sample of $X$, is a shrinking region $[\hat{ X}_t-C_t, \hat{ X}_t+C_t]$ around $\widehat{X}_t$ that reflects where \samp is almost certain that the true value of $\mE X$ has to be. We leave the exact definition of both $C_t$ and $N_1$ until Theorems \ref{thm:sampling} and \ref{thm:sampling2} for reasons that will become clear. Once the lower bound of the confidence region crosses $w-\epsilon$, or the upper bound crosses $w+\epsilon$, \samp completes and returns true or false respectively. \samp stops sampling in at most $N_1$ samples regardless, and in this case returns true or false according to whether $\hat{ X}_t\geq w$. An illustration of the various states of \samp is depicted in Figure \ref{fig:samp}, and pseudocode for \samp is provided in Algorithm \ref{alg:samp}.

\begin{figure}[t]
\centering
\begin{tikzpicture}[scale=0.7]

\draw (0,0) to (9,0);
\draw node at (10,1) {$w+\epsilon$};
\draw node at (10,0) {$w$};
\draw node at (10,-1) {$w-\epsilon$};
\draw[thick,dotted,red] (0,1) to (9,1);
\draw[thick,dotted,red] (0,-1) to (9,-1);

\draw[thick,cyan,->] (1,0.5) to (1,2);
\draw[thick,cyan,->] (1,0.5) to (1,-1);
\draw[cyan,fill=cyan] (1,0.5) circle [radius=0.1];
\draw node at (1.5,0.5) {$(a)$};

\draw[thick,cyan,->] (3,-0.5) to (3,1);
\draw[thick,cyan,->] (3,-0.5) to (3,-2);
\draw[cyan,fill=cyan] (3,-0.5) circle [radius=0.1];
\draw node at (3.5,-0.5) {$(b)$};

\draw[thick,cyan,->] (5,0.5) to (5,2.5);
\draw[thick,cyan,->] (5,0.5) to (5,-1.5);
\draw[cyan,fill=cyan] (5,0.5) circle [radius=0.1];
\draw node at (5.5,0.5) {$(c)$};

\draw[thick,cyan,->] (7,0) to (7,1);
\draw[thick,cyan,->] (7,0) to (7,-1);
\draw[cyan,fill=cyan] (7,0) circle [radius=0.1];
\draw node at (7.5,0.5) {$(d)$};




\end{tikzpicture}
\caption{An illustration of the various states of \samp. The blue dots depict the values of $\widehat{X}_t$, while the surrounding blue lines depict the confidence region $[\widehat{X}_t-C_t, \widehat{X}_t+C_t]$. Once the region looks like (a), \samp will return true. In (b), \samp will return false. In (c), \samp will continue sampling to reduce the width of the confidence region. Finally, in (d) \samp has taken $N_1$ samples resulting in an $\epsilon$-additive approximation.}
\label{fig:samp}
\end{figure}
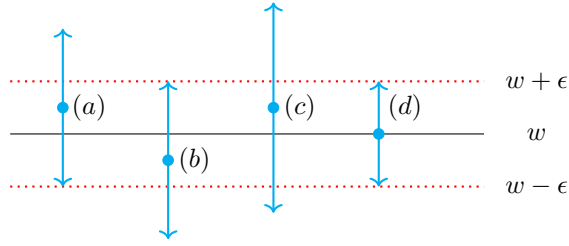

\begin{algorithm}[t]
\caption{\samplong (\samp)}
\label{alg:samp}
\begin{algorithmic}[1]
    \STATE \textbf{Input:} $w$, $\epsilon$, $\delta$, $\mathcal{D}_X$, $R$
    \FOR{$t=1,2,...N_1$}\label{line:sample_N_1}
        \STATE $\widehat{X}_t\gets$ updated sample mean after taking $t$-th sample from $\mathcal{D}_X$
        \STATE $\conf\gets$ updated confidence interval \label{alg: update confidence interval}
        \IF{$\widehat{X}_t-\conf \geq w-\epsilon$} 
        \label{line: comparison to thres 1}
            \STATE \textbf{return true}
        \ELSIF{$\widehat{X}_t+\conf \leq w+\epsilon$} 
        \label{line: comparison to thres 2}
            \STATE \textbf{return false}
        \ENDIF
 \ENDFOR
 \IF{$\widehat{X}_t\geq w$}
 \STATE \textbf{return true}
 \ELSE
 \STATE \textbf{return false}
 \ENDIF
\end{algorithmic}
\end{algorithm}

We now state our first main result for \samp in Theorem \ref{thm:sampling} below. The second item of Theorem \ref{thm:sampling} states that with high probability, \samp will correctly return the answer to whether $\mE X$ is approximately above or below the input threshold $w$. The first item states that, in the worst case, \samp takes $O(R^2\log(1/\delta)/\epsilon^2)$ samples from $\mathcal{D}_X$ to return true or false no matter what the value of $\mE X$ is (the first of the values we take the minimum of). However, the further the value of $\mE X$ is from $w$, as reflected by $\phi$, the fewer samples \samp needs to make a decision (the second of the values that we take the minimum of). Figure \ref{fig:sample_complexity} illustrates how the sample complexity changes with the increase of gap function $\phi$ in the result of Theorem \ref{thm:sampling}. The details of the proof of Theorem \ref{thm:sampling} can be found in Section \ref{appdx:proof_of_samp} of the supplementary material.
\begin{theorem}
    \label{thm:sampling}
   For any random variable $X$ that is $R$-sub-Gaussian, if we define $N_1=2R^2/\epsilon^2\log \frac{4}{\delta}$, and $\conf =R\sqrt{\frac{2}{t}\log \frac{8 t^2}{\delta}}$, then the algorithm \samplong achieves that with probability at least $1-\delta$
    \begin{enumerate}[noitemsep]
        \item \samp on input $(w,\epsilon,\delta,\mathcal{D}_X,R)$ takes at most the minimum between
    \begin{align*}
        \left\{\frac{2R^2}{\epsilon^2}\log \left(\frac{4}{\delta}\right),\frac{8R^2}{\phi_X^2}\log\left(\frac{16R^2}{\phi_X^2}\sqrt{\frac{2}{\delta}}\right)\right\}
    \end{align*}
    noisy samples, where $\phi_X = \frac{\epsilon + |w-\mathbb{E} X|}{2}$.
    \item If \samp returns true, then $\mE X\geq w-\epsilon$. If \samp returns false, then $\mE X\leq w+\epsilon$.
    \end{enumerate}
\end{theorem}

\begin{figure}[t!]
\hspace{-0.5em}
\centering
\includegraphics[width=0.8\columnwidth]{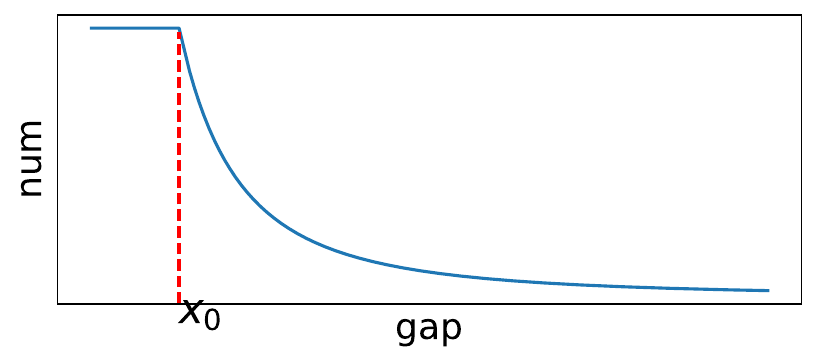}
\captionof{figure}{A plot to illustrate how the number of samples taken by \samp (num) changes with the gap function $\phi_X$ (see Theorem \ref{thm:sampling}). There exists some $x_0$ such that when $0<\phi_X\leq x_0$, the required number of samples is $\frac{R^2}{2\epsilon^2}\log \frac{2}{\delta}$ (the left side in the sample complexity result in Theorem \ref{thm:sampling}).
When $\phi_X>x_0$, the right-hand side in Theorem \ref{thm:sampling} is the minimum and the sample complexity of the algorithm decreases fast as $\phi_X$ increases. 
}
        \label{fig:sample_complexity}
\end{figure}

Our second result, Theorem \ref{thm:sampling2}, is related to Theorem \ref{thm:sampling} but instead of an additive approximation error (i.e. $\mE X\geq w-\epsilon$ or $\mE X\leq w+\epsilon$), the error is a combination of multiplicative and additive. This alternative result is useful in some algorithms such as that considered in Section \ref{appdx:mono2} in the appendix and Section \ref{sec:matroid}. In order to get Theorem \ref{thm:sampling2}, a different definition of the confidence radius $C_t$ as well as the maximum number of samples $N_1$ is needed. Theorem \ref{thm:sampling2} is proven in the supplementary material in Section \ref{appdx:proof_of_samp2}.
\begin{theorem}
\label{thm:sampling2}
    For any random variable $X$ that is bounded in the range of $[0,R]$, if we define $C_t=\frac{3R}{t\alpha}\log(\frac{8R^2}{\delta})$, and $N_1=\frac{3R}{\alpha\epsilon}\log(\frac{4}{\delta})$ where $\alpha$ is an additional parameter that controls the multiplicative error rate, the algorithm \samplong achieves that with probability at least $1-\delta$
    \vspace{-2mm}
    \begin{enumerate}[noitemsep]
        \item \samp on input $(w,\epsilon,\delta,\mathcal{D}_X,R)$ takes at most the minimum between
    \begin{align*}
        \left\{\frac{3R}{\epsilon\alpha}\log \left(\frac{4}{\delta}\right),\frac{12R}{\alpha\phi_X'}\log\left(\frac{12R}{\alpha\phi_X'}\sqrt{\frac{8}{\delta}}\right)\right\}
    \end{align*}
    noisy samples, $\phi_X' = \frac{\epsilon -\alpha\mE X+| w-\mathbb{E} X|}{2}$.
    \item If the output is true, then $(1+\alpha)\mE X\geq w-\epsilon$. If the output is false, then $(1-\alpha)\mE X\leq w+\epsilon$.
    \end{enumerate}
\end{theorem}

\section{Monotone Submodular Maximization}
\label{sec:monotone}
In this section, we address the MSMC problem under the noisy setting, where we assume the noisy sampling of the marginal gain ${\Delta f}(S,s)$ is $R$-sub-Gaussian for any $S\subseteq U$ and $s\in U$. 
Necessary definitions and notations are first given in Section \ref{sec:prelim}. We propose two algorithms \alglong (\alg) and \alglongmono (\algmono) for this problem. A detailed description of \alg is given in Section \ref{sect:alg}. The approximation and sample
complexity guarantees of \alg are presented in Theorem \ref{mainthm} in Section \ref{sect:analysis}. For \algmono, the algorithm description is provided in Algorithm \ref{alg:ATG2}, and the theoretical results are presented in Section \ref{appdx:mono2} in the appendix.

\subsection{Algorithm Description}
\label{sect:alg}

We now describe our proposed algorithm, \alglong (\alg). \alg is based on the algorithm \thresholdlong (\threshold) of \citet{badanidiyuru2014fast} which is for \prob with an exact value oracle. Pseudocode for \alg can be found in Algorithm \ref{alg:ATG}. 

The algorithm \alg takes as input a parameter $\alpha \in (0,1)$. At all times throughout \alg, there is a marginal gain threshold value $w$ (that decreases over time) and a partial solution $S$ (to which elements are iteratively added). \alg proceeds in $O(\log(\kappa/\alpha)/\alpha)$ \textit{rounds}, where each round corresponds to a value of $w$ and a pass through $U$. During each round, \alg iterates through all elements in $U$. Since for each $S$ and $u$, the noisy query to the marginal gain ${\Delta f}(S,u)$ is $R$-sub-Gaussian, \alg can use \samp as the subroutine to determine whether to include $u$ to the solution set $S$. More specifically, \alg calls the procedure \samplong (\samp) with threshold $w$, approximation error bound $\epsilon$, error probability $\frac{2\delta}{3nh(\alpha)}$ where $h(\alpha)=\frac{\log{(\kappa/\alpha)}}{\alpha}$, random distribution $\mathcal{D}(S,u)$, and sub-Gaussian parameter $R$ as input. The worst-case query complexity $N_1$ and confidence interval $C_t$ in \samp are defined as in Theorem \ref{thm:sampling}.

The threshold $w$ is first set to $d$, which is an $\epsilon$-additive approximation of the maximum singleton value with high probability. In particular, $d$ satisfies that with probability at least $1-\delta/3$, $\max_{s\in U}f(s)+\epsilon\geq d\geq \max_{s\in U}f(s)-\epsilon$. At the end of each round, $w$ is decreased by a factor of $1-\alpha$. \threshold completes once $w$ reaches $\alpha d/\kappa$, or if $S$ has reached the cardinality constraint $\kappa$, whichever comes first. 

\begin{algorithm}[t]
\caption{\alglong (\alg)}\label{alg:ATG}
 \begin{algorithmic}[1]
 \STATE \textbf{Input:} $\epsilon$, $\delta, \alpha$
 \STATE $N_2\gets 2R^2\log(6n/\delta)/(\epsilon^2)$
 \FORALL{$s\in U$}
 \STATE $\hat{f}(s) \gets $ sample mean over $N_2$ samples from $\mathcal{D}(\emptyset,s)$ \label{alg:ATG:line:sample-mean}
 \ENDFOR
  
  \STATE $d:=\max_{s\in U}\hat{f}(s)$, 
 \STATE $w\gets d$, $S\gets \emptyset$
 \WHILE{$w>\alpha d/\kappa$}
 \FORALL{$u\in U$} 
\IF{$|S|<\kappa$}
 \STATE thre = \samplong($w$, $\epsilon$, $\frac{2\delta}{3nh(\alpha)}$, $\mathcal{D}(S,u)$, $R$)
 \IF{thre}
 \STATE $S\gets S\cup \{u\}$
 \ENDIF
  \ENDIF
 \ENDFOR

 \STATE $w=w(1-\alpha)$
 \ENDWHILE
 \STATE \textbf{return} $S$
 \end{algorithmic}
\end{algorithm}


\subsection{Theoretical Guarantee}
\label{sect:analysis}
Our main result is the Theorem \ref{mainthm} below.
\begin{theorem}
    \label{mainthm}
    Suppose the noisy marginal gain of any subset $S\subseteq U$ and element $s\in U$ is $R$-sub-Gaussian, then \alg makes at most $n\log(\kappa/\alpha)/\alpha$ calls of \samp. In addition, with probability at least $1-\delta$, the following statements hold:
    \begin{itemize}[noitemsep]
        \item The exact function value of the output solution set $S$ satisfies that $f(S)\geq(1-e^{-1}-\alpha)f(OPT)-2\kappa\epsilon$;
    \item Each call of \samp on input ($w$, $\epsilon$, $\frac{2\delta}{3nh(\alpha)}$, $\mathcal{D}(S,u)$, $R$) takes at most the minimum between
    \begin{align*}
        \frac{8R^2}{\phi^2(S,u)}\log\left(\frac{16R^2\sqrt{\frac{3nh(\alpha)}{\delta}}}{\phi^2(S,u)}\right)
    \end{align*}
    and
    \begin{align*}
       \frac{2R^2}{\epsilon^2}\log \left(\frac{6nh(\alpha)}{\delta}\right)
    \end{align*}
    noisy samples. Here $OPT$ is an optimal solution to the MSMC problem, $$\phi(S,u) = \frac{\epsilon + |w-\Delta f(S,u)|}{2},$$ and $$h(\alpha)=\frac{\log{(\kappa/\alpha)}}{\alpha}.$$
    \end{itemize}
    
\end{theorem}
As one can see, the approximation guarantee of Theorem \ref{mainthm} can be made arbitrarily close to the best possible $1-1/e$ by choosing $\alpha$ and $\epsilon$ arbitrarily small. On the other hand, more samples will be needed as a tradeoff for increased accuracy. 

The proof and analysis of Theorem \ref{mainthm} are deferred to Section \ref{appdx:proof_of_mono} in the appendix. Additionally, a comparison of the theoretical guarantees between our results and those of \singla in \cite{singla2016noisy} is provided in Section \ref{app:related} in the appendix.

\section{Non-monotone Submodular Objectives}
\label{sec:nonmono}



In Section \ref{sec:monotone} and Section \ref{sec:matroid}, we employ the adaptive sampling algorithm \samp as a subroutine in algorithms that share the same intuition as \threshold to determine if the marginal gain is approximately above or below the threshold $w$. In this section, we demonstrate that \samp can also be employed to develop a deterministic algorithm for the Submodular Maximization (USM) problem, following a similar idea as in \cite{buchbinder2015tight}. Here we assume that the sampling of the marginal gain $\Delta f(S,s)$ is $R$-sub-Gaussian for any $S\subseteq U$ and $s\in U$. 

We propose the algorithm \texttt{CDG}, which is based upon the deterministic algorithm presented in \cite{buchbinder2015tight} ("Double Greedy") for USM in the noise-free setting, with our procedure \samp integrated into it in order to deal with the noisy access to $f$. Here the parameters $N_1$ and $C_t$ in the subroutine algorithm \samp are defined in accordance with Theorem \ref{thm:sampling}.
 We denote the sets $A$ and $B$ after the $i$-th iteration in \cdg as $A_i$ and $B_i$, and the element processed in the $i$-th iteration as $u_i$. Pseudocode for \texttt{CDG} is presented in Algorithm \ref{alg:CDG}.

We start by briefly describing the deterministic algorithm in \cite{buchbinder2015tight}. In particular, the algorithm of \cite{buchbinder2015tight} maintains two sets $A$ and $B$ as it makes a single pass through the ground set $U$ in the order $u_1,...,u_n$. At each element $u_i$, the algorithm evaluates whether $\Delta f(A_{i-1}, u_i)$, the marginal gain of adding the new element $u_i$, surpasses the loss incurred by removing it from set $B_{i-1}/\{u_i\}$, which is $-\Delta f(B_{i-1}/\{u_i\}, u_i)$. If $\Delta f(A_{i-1}, u_i)\geq-\Delta f(B_{i-1}/\{u_i\}, u_i)$, then $u_i$ is added to the final solution set. Otherwise, it is removed from $B_{i-1}$. Our insight is that this procedure in fact is asking about whether the value of the function $\Delta f(A_{i-1}, u_i)+\Delta f(B_{i-1}/\{u_i\}, u_i)$ is above or below the threshold $0$. 

It is important to note that \samp cannot be used as a subroutine in the randomized algorithm with a $1/2$ approximation guarantee as presented in \cite{buchbinder2015tight}. This is due to a fundamental difference in the requirements of the two algorithms. The randomized algorithm in \cite{buchbinder2015tight} requires knowing the exact ratio of $\frac{\Delta f(A_{i-1}, u_i)}{\Delta f(A_{i-1}, u_i)+\Delta f(B_{i-1}/\{u_i\}, u_i)}$, while \samp only guarantees the difference between the mean of a random variable and a threshold value $w$.
Therefore, in the deterministic algorithm, we can apply \samp to find whether the expectation of $X_i=\widetilde{\Delta f}(A_{i-1},u_i)+\widetilde{\Delta f}(B_{i-1}/\{u_i\},u_i)$ is approximately above or below $0$.


We now present our theoretical guarantees for \texttt{CDG} below in Theorem \ref{thm:nonmono}. The proof of Theorem \ref{thm:nonmono} can be found in the supplementary material. We note that our algorithm \texttt{CDG} achieves nearly the same approximation guarantee as that of \cite{buchbinder2015tight}, but with a small penalty due to the noisy setting.
\begin{theorem}
    \label{thm:nonmono}
    \cdg makes $n$ calls of \samp. In addition, with probability at least $1-\delta$, the following statements hold:
    \begin{enumerate}
    [noitemsep]
        \item The exact function value of the output solution set $S$ satisfies that $f(S)\geq\frac{f(OPT)}{3}-\epsilon$;
        \item Each call of \samp on input $(0,\frac{3\epsilon}{n},\frac{\delta}{n},\mathcal{D}_{X_i},\sqrt{2}R)$ takes at most the minimum between
    \begin{align*}
        \left\{\frac{4n^2R^2}{9\epsilon^2}\log \left(\frac{4n}{\delta}\right),\frac{16R^2}{\phi^2_{i}}\log\left(\frac{32R^2}{\phi^2_{i}}\sqrt{\frac{2n}{\delta}}\right)\right\}
    \end{align*}
    noisy samples. Here $OPT$ is an optimal solution to the USM problem, and
    \begin{align*}
        \phi_{i} &:= \frac{3\epsilon/n + |\mE X_i|}{2}\\
        &=\frac{3\epsilon/n + |\Delta f(A_{i-1},u_i)+\Delta f(B_{i-1}/\{u_i\},u_i)|}{2}. 
    \end{align*}
    \end{enumerate}
   
\end{theorem}

From Theorem \ref{thm:nonmono}, we can see that \cdg achieves an approximation guarantee that is arbitrarily close to $1/3$, which matches the result of the deterministic algorithm in \cite{buchbinder2015tight}.

\begin{algorithm}[t]
\caption{\texttt{Confident Double Greedy} (\texttt{CDG} )}\label{alg:CDG}
 \begin{algorithmic}[1]
 \STATE \textbf{Input:} $\epsilon$, $\delta$
 \STATE  $A\gets \emptyset$, $B\gets U$
 
 \FOR{\textbf{all} $u\in U$}
 \STATE Define r.v. $X=\widetilde{\Delta f}(A,u)+\widetilde{\Delta f}(B/\{u\},u)$, 
 \STATE thre = \samplong($0$, $\frac{3\epsilon}{n}$, $\frac{\delta}{n}$, $\mathcal{D}_{X}$, $\sqrt{2}R$)
 \IF{$thre$} 
 \STATE $A\gets A\cup\{u\}$
 \ELSE \STATE  $B\gets B/\{u\}$
 \ENDIF
 \ENDFOR
 \STATE \textbf{return} $A$
 \end{algorithmic}
\end{algorithm}
\section{Continuous Threshold Greedy with Noisy Queries}
\label{sec:matroid}
Finally, in this section, we consider the problem of Monotone Submodular Maximization with a Matroid constraint (MSMM). We propose the \contialglong (\contialg) algorithm for MSMM, which leverages the continuous multilinear extension $\vect{F}$ of the submodular function $f$ to obtain an approximation guarantee arbitrarily close to the best possible result of $1-1/e$.

In many applications, even with access to an exact oracle for $f$, $\vect{F}$ is not able to be evaluated exactly due to the inherent randomness in $S(\vect{x})$ in the definition of $\vect{F}$ (see Section \ref{sec:prelim}), so we can only make noisy queries to $\vect{F}$. In addition, our results hold even for the case that only noisy access to $f$ is provided. More specifically,  we assume that for any set $S\subseteq U$ and element $s\in U$, the noisy marginal gain $\widetilde{\Delta f}(S,s)$ is bounded in $[0,R]$. 


We now describe \contialg. Let $\kappa$ to denote the rank of the matroid, and let $S(\vect{x})$ be a random set that contains each element $i\in U$ with probability $x_i$
The \contialg algorithm initializes a solution in the origin, $\vect{x}=\vect{0}$. Then at each step, \contialg selects a subset of coordinates $B$ to increment by a predetermined step size $\epsilon$.  The set of coordinates $B$ is chosen by the subroutine algorithm \contisublong (\contisub), which is described in Algorithm \ref{alg:CCTG_subroutine}. Here the parameters $N_1$ and $C_t$ in the subroutine algorithm \samp are defined as in Theorem \ref{thm:sampling2} with the multiplicative error parameter $\alpha$ set to be $\epsilon/3$. After the \contialg is complete, we process the fractional solution $\vect{x}$ with the swap rounding procedure in \cite{vondrak2011submodular} to obtain the final solution set $S$. 


\begin{algorithm}[t]
\caption{\contialglong(\contialg)}\label{alg:CCTG}
 \begin{algorithmic}[1]
 \STATE \textbf{Input:} $\epsilon$, $\delta$, $\mathcal{M}\in 2^U$
 \STATE  $\vect{x}\gets \vect{0}$
 \FORALL{$s\in U$ and $s\in\mathcal{M}$}
 \STATE $\hat{f}(s) \gets $ sample mean over $\frac{18\kappa}{\epsilon^2}\log\frac{4n}{\delta}$ samples from $\mathcal{D}(\emptyset,s)$
 \ENDFOR
  
  \STATE $d:=\max_{s\in \mathcal{M}}\hat{f}(s)$, 
 
 \FOR{$t=1$ to $1/\epsilon$}
 \STATE $B\gets$\contisublong($\vect{x}$, $\epsilon$, $\delta$, $d$, $\mathcal{M}$)
 \STATE $\vect{x}\gets \vect{x}+\epsilon\cdot\vect{1}_{B}$
 \ENDFOR
 \STATE \textbf{return} $\vect{x}$
 \end{algorithmic}
\end{algorithm}

\begin{algorithm}[t]
\caption{\contisublong(\contisub)}\label{alg:CCTG_subroutine}
 \begin{algorithmic}[1]
 \STATE \textbf{Input:} $\vect{x}$, $\epsilon$, $\delta$, $d$, $\mathcal{M}\in 2^U$
 \STATE $w\gets d$, $B\gets \emptyset$
 \WHILE{$w>\frac{\epsilon d}{3\kappa}$}
 \FORALL{$u\in U$} 
\IF{$B\cup\{u\}\in\mathcal{M}$}
\STATE $X=\widetilde{\Delta f}(S(\vect{x}+\epsilon\vect{1}_B),u)$
 \STATE thre = \sampnewlong($w$, $\frac{R\epsilon}{2\kappa}$, $\frac{\delta\epsilon}{2nh'(\epsilon)}$, $\mathcal{D}_X$, $R$)
 \IF{thre}
 \STATE $B\gets B\cup \{u\}$
 \ENDIF
  \ENDIF
 \ENDFOR

 \STATE $w=w(1-\epsilon/3)$
 \ENDWHILE
 \STATE \textbf{return} $B$ \end{algorithmic}
\end{algorithm}

\begin{theorem}
\label{thm:continuous}
    \contialg makes at most $\frac{3n}{\epsilon^2}\log\frac{3\kappa}{\epsilon}$ calls of \samp. In addition, with probability at least $1-\delta$, the following statements hold:
    \begin{itemize}[noitemsep]
        \item The output fractional solution $\vect{x}$ achieves the approximation guarantee of $\vect{F}(\vect{x})\geq(1-e^{-1}-2\epsilon)f(OPT)-R\epsilon$.
        \item Each call of \sampnew on input  ($w$, $\frac{\epsilon R}{2\kappa}$, $\frac{\delta\epsilon}{2nh'(\epsilon)}$, $\mathcal{D}_X$, $R$) requires at most the minimum between
    \begin{align*}
        \left\{\frac{18\kappa}{\epsilon^2}\log \left(\frac{8nh'(\epsilon)}{\delta\epsilon}\right),\frac{36R}{\epsilon\phi''_X}\log\left(\frac{144R}{\epsilon\phi''_X}\sqrt{\frac{nh'(\epsilon)}{\delta\epsilon}}\right)\right\}
    \end{align*}
    noisy queries to the marginal gain. Here $OPT$ is an optimal solution to the MSMM problem, and
    \begin{align*}
        \phi_X'' = \frac{\frac{\epsilon R}{2\kappa} -\epsilon\mathbb{E}X  /{3}+ |w-\mathbb{E}X|}{2},
    \end{align*} and $h'(\epsilon)=\frac{3}{\epsilon}\log{(\frac{3\kappa}{\epsilon})}$.
    
    \end{itemize}
\end{theorem}
The proof of Theorem \ref{thm:continuous} is deterred to Appendix \ref{appdx:continuous}.
Besides, we discuss and compare our results in Theorem \ref{thm:continuous} with the Accelerated Continuous Greedy algorithm in \cite{badanidiyuru2014fast} in Section \ref{appdx:comparison_to_ACG}.

\section{Applications and Experiments}
\label{sec:exp}
\begin{figure*}[t!]
    \centering
    \hspace{-0.5em}
     \subfigure[delicious\_300 samples]
{\label{fig:cover2500_300_eps_q}\includegraphics[width=0.24\textwidth]{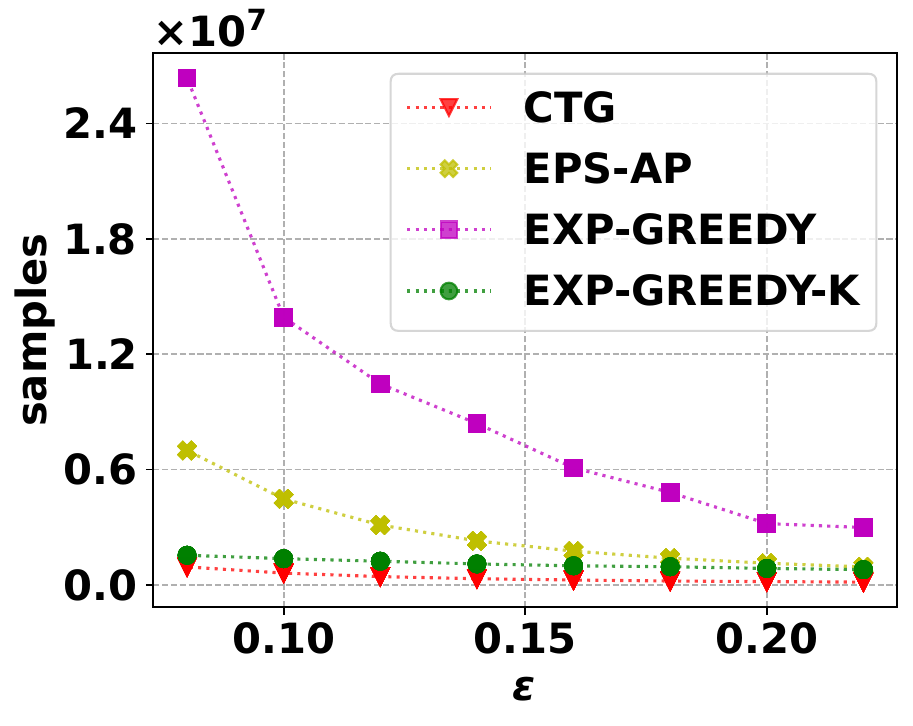}} 
\hspace{-0.5em}
     \subfigure[delicious\_300 average samples]
{\label{fig:cover2500_300_eps_average-q}\includegraphics[width=0.24\textwidth]{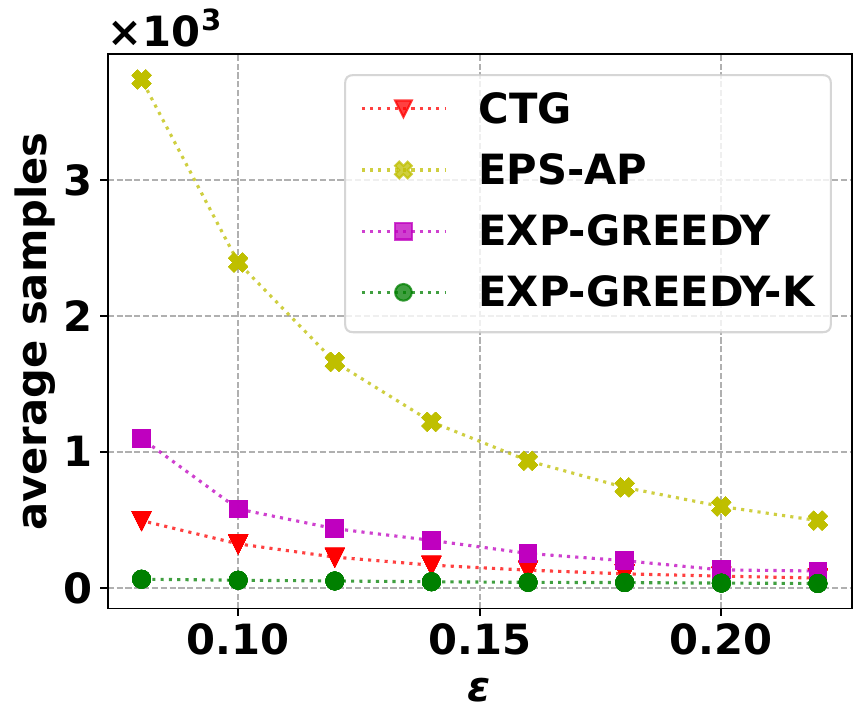}}
    \hspace{-0.5em}
     \subfigure[delicious\_300 samples]
{\label{fig:cover_300_k_q}\includegraphics[width=0.24\textwidth]{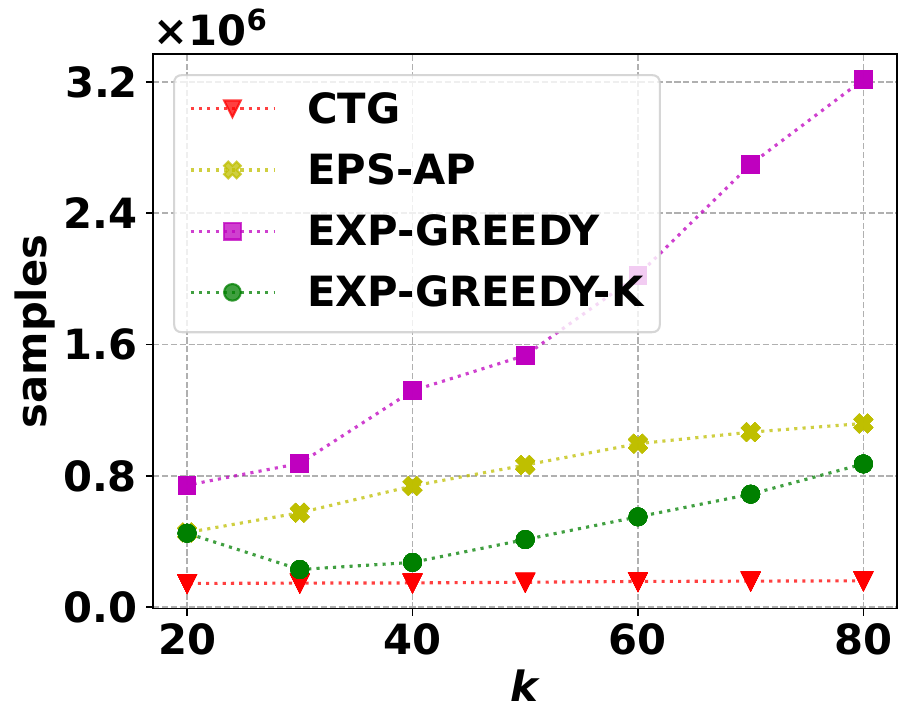}}
\hspace{-0.5em}
\subfigure[delicious\_300 average samples]
{\label{fig:cover_300_k-ave_q}\includegraphics[width=0.24\textwidth]{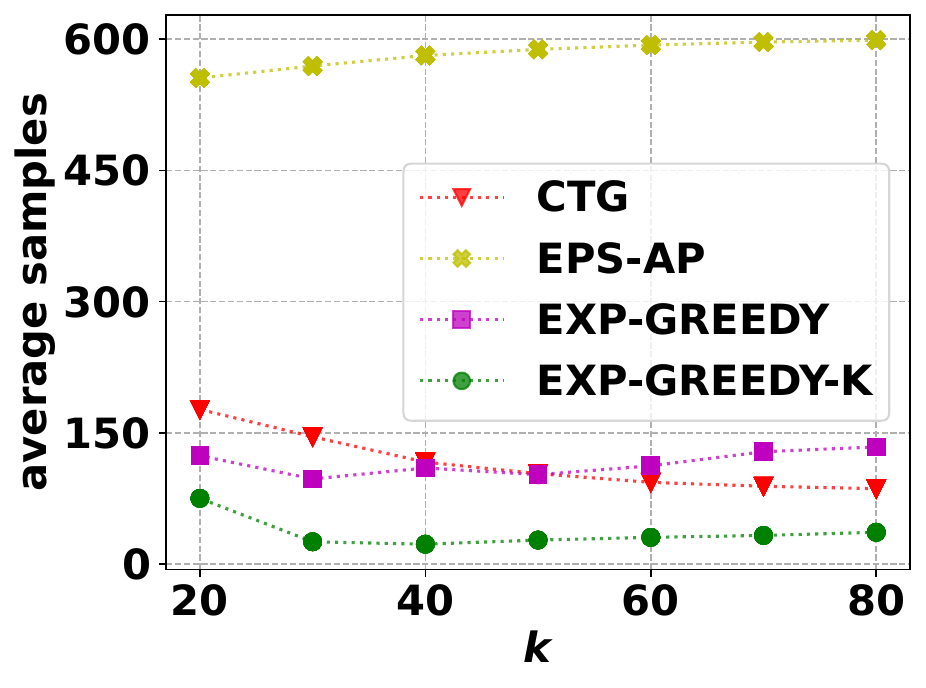}}
\hspace{-0.5em}
\subfigure[corel\_60 samples]
{\label{fig:corel_60_eps_q}\includegraphics[width=0.24\textwidth]{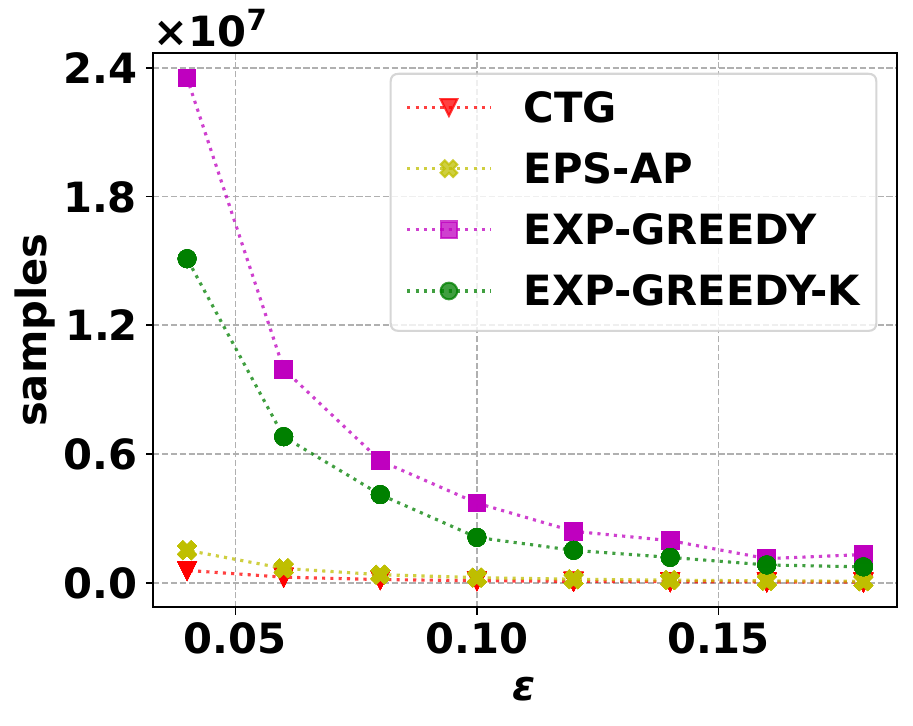}}
\hspace{-0.5em}
\subfigure[corel\_60 average samples]
{\label{fig:corel_60_eps_ave_q}\includegraphics[width=0.24\textwidth]{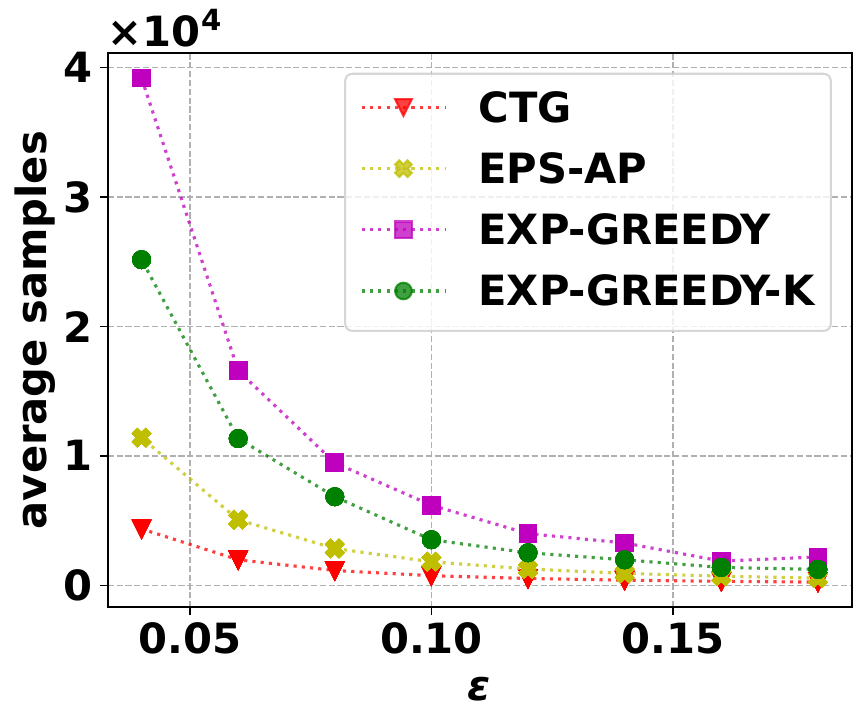}}
\hspace{-0.5em}
\subfigure[corel\_60 samples]
{\label{fig:corel_60_k-q}\includegraphics[width=0.24\textwidth]{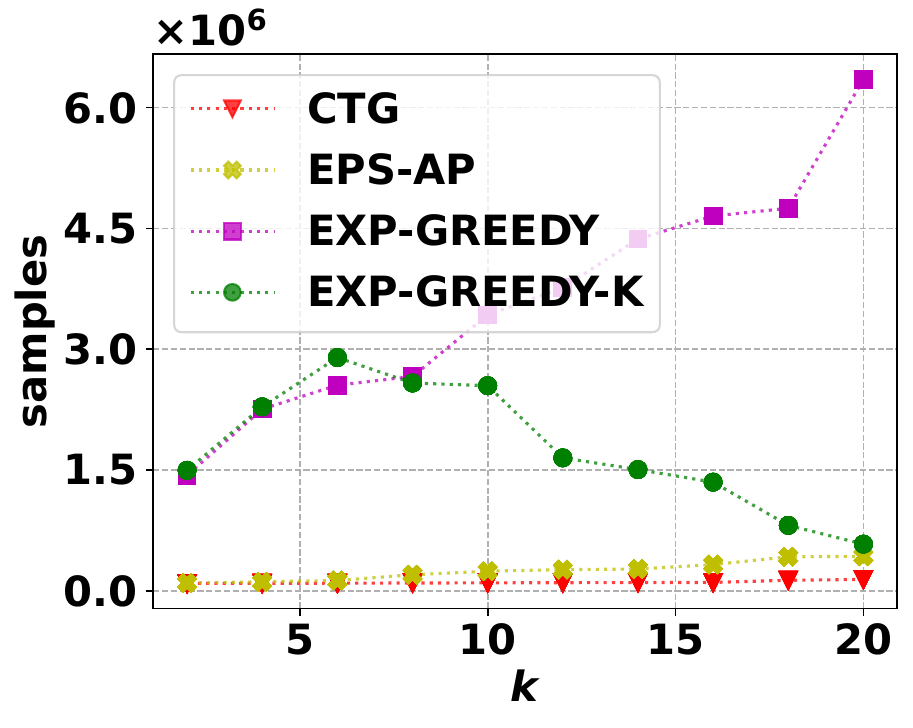}}
\hspace{-0.5em}
\subfigure[corel\_60 average samples]
{\label{fig:corel_60_k_ave_q}\includegraphics[width=0.24\textwidth]{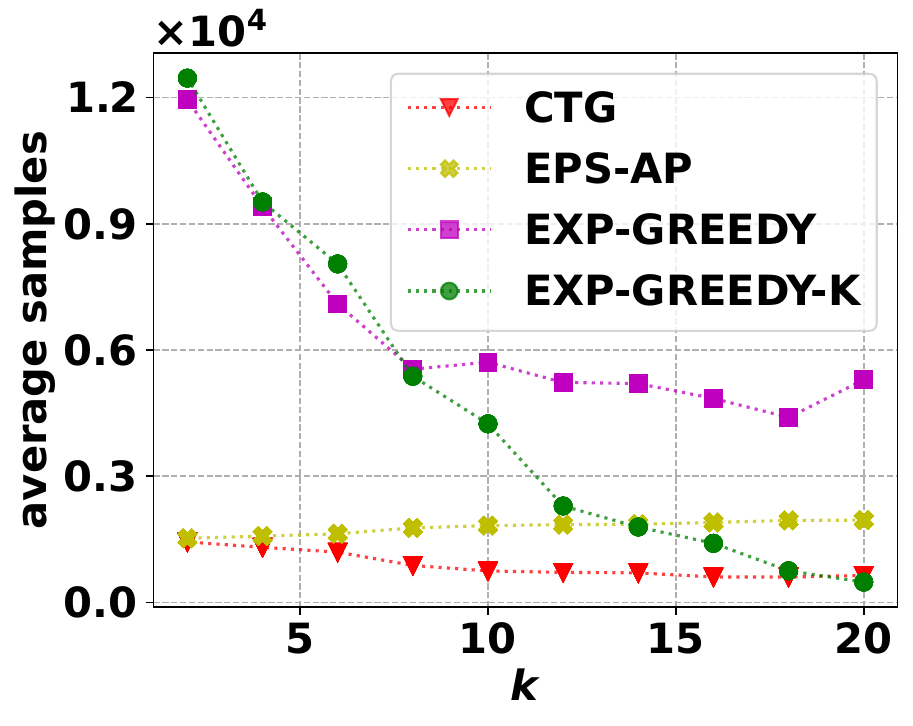}}
\hspace{-0.5em}
\subfigure[delicious samples]
{\label{fig:cover_eps_q}\includegraphics[width=0.24\textwidth]{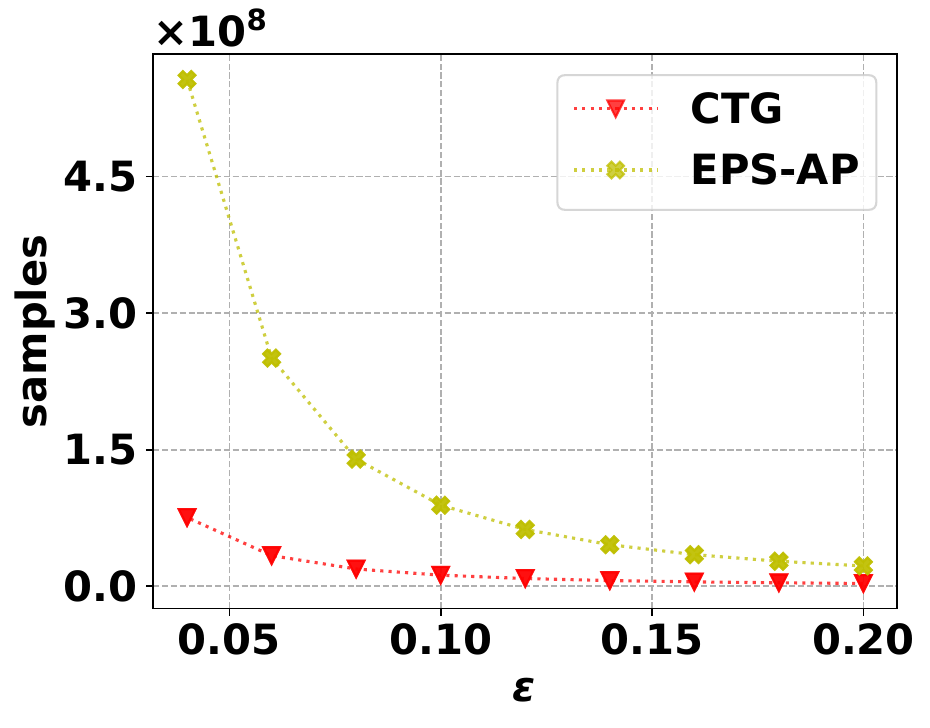}}
\hspace{-0.5em}
\subfigure[delicious samples]
{\label{fig:cover_k_q}\includegraphics[width=0.24\textwidth]{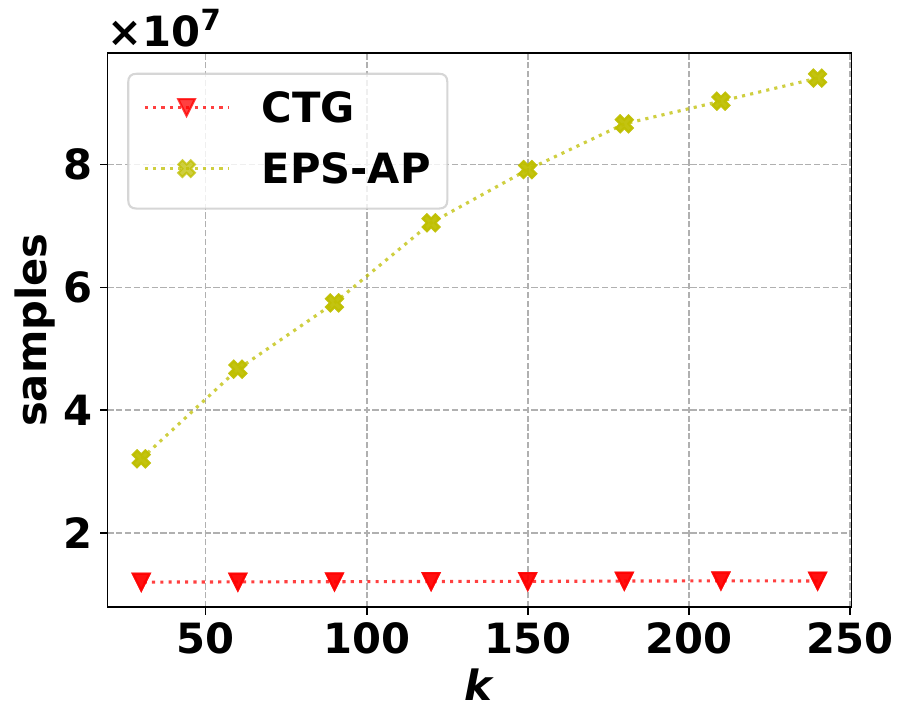}}
\hspace{-0.5em}
\subfigure[corel samples]
{\label{fig:corel_eps-q}\includegraphics[width=0.24\textwidth]{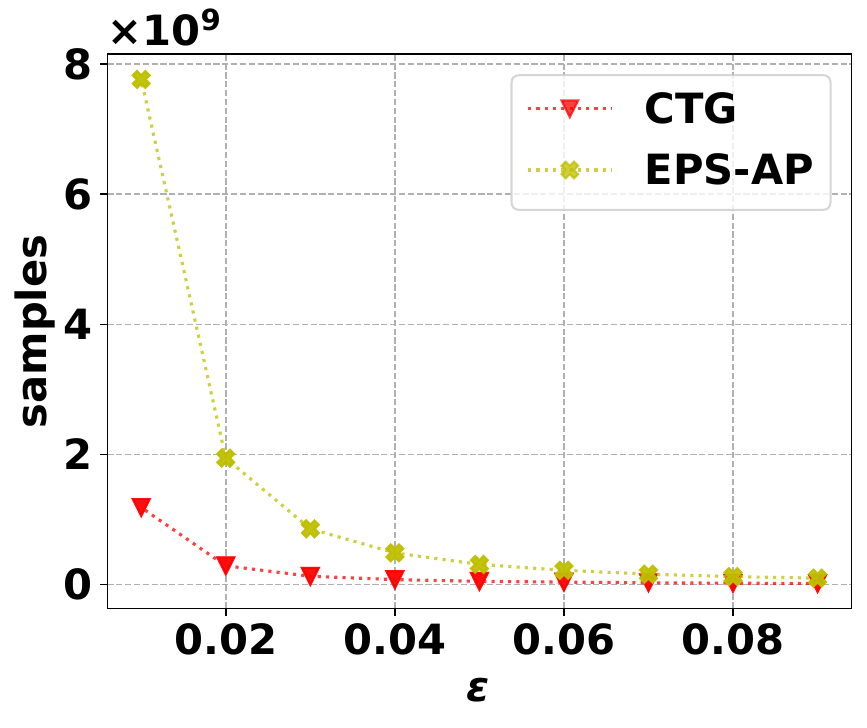}}
\hspace{-0.5em}
\subfigure[corel samples]
{\label{fig:corel_eps_ave_q}\includegraphics[width=0.24\textwidth]{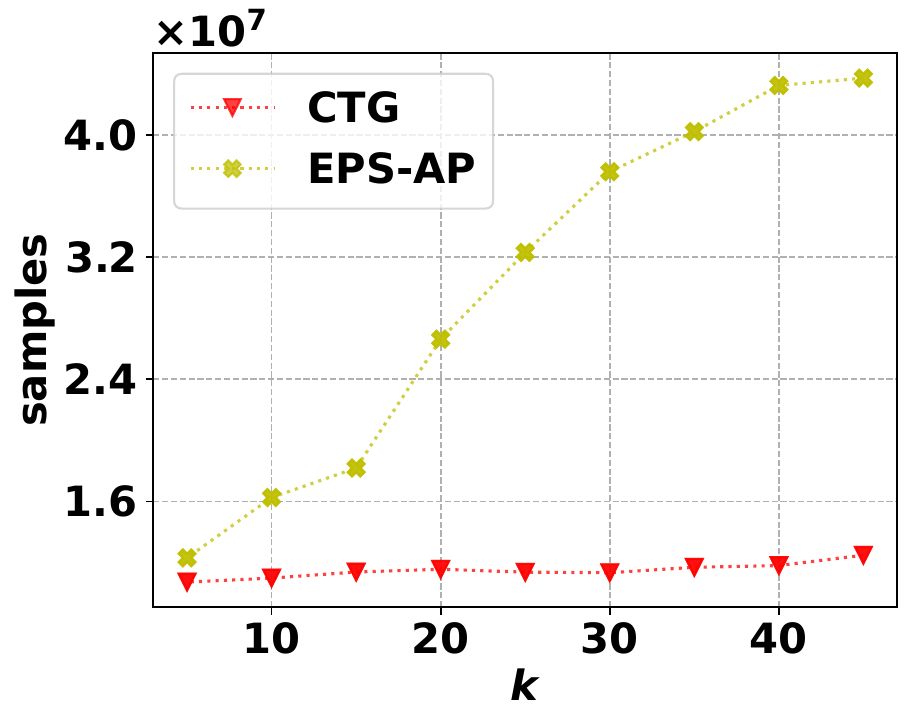}}

\caption{The experimental results of running different algorithms on instances of data summarization on the delicious URL dataset ("delicious", "delicious\_300") and Corel5k dataset ("corel", "corel\_60").}
\label{fig:exp_results}
\end{figure*}
In this section, we conduct an experimental evaluation of our algorithm \alg on instances of \prob with noisy marginal gain evaluations. In particular, we consider instances of the noisy data summarization application, which is described in Section \ref{sec:data_summarization} in the appendix. Synthetic noise is introduced into marginal gain queries by adding a zero-mean Gaussian random variable with $\sigma=1.0$ ($\sigma$ is the standard deviation) to the exact value of marginal gain. Therefore, parameter $R=1.0$. Our experiments are conducted on a subset of the Delicious dataset of URLs that are tagged with topics \cite{soleimani2016semi}, and subsets of the Corel5k dataset of tagged images \cite{duygulu2002object}. We give more details about the datasets we use in the appendix in the supplementary material. We additionally consider the influence maximization problem in the appendix in the supplementary material. The setup of our experiments is described in Section \ref{sec:setup}, while our results are presented in Section \ref{sec:exp_results}.

\subsection{Experimental Setup}
\label{sec:setup}
We now describe the setup of our experiments. In addition to our algorithm \alg, we compare the following alternative approaches to noisy \prob: (i) The fixed $\epsilon$ approximation (``\texttt{EPS-AP}'') algorithm; (ii)  Two special case of the algorithm \singla of \citet{singla2016noisy} ``\texttt{EXP-GREEDY}'' and  ``\texttt{EXP-GREEDY-K}'' with the parameter $k'$ in \singla set to be $k'=1$ and  $k'=\kappa$ respectively. More details about the three algorithms can be found in the appendix.
 We evaluate \alg and \texttt{EPS-AP} on all the datasets. However, \texttt{EXP-GREEDY} and \texttt{EXP-GREEDY-K} have greater runtime as discussed in the appendix in the supplementary material, and so we only evaluate them on the smaller datasets.
 Details about the parameter settings can be found in the appendix in the supplementary material.

\subsection{Experimental Results}
\label{sec:exp_results}

We now present our experimental results. The algorithms are compared in terms of: (i) The function value $f$ of their solution; (ii) The total number of noisy samples of the marginal gain; (iii) The average number of samples per marginal gain estimation \textit{average samples=}$$\textit{total samples}/\#\textit{ of evaluated marginal gains}.$$

Our results for different values of $\epsilon$ and $\kappa$ are presented in Figure \ref{fig:exp_results}. From Figures \ref{fig:cover2500_300_eps_q}, \ref{fig:cover_300_k_q}, \ref{fig:corel_60_eps_q} and \ref{fig:corel_60_k-q}, one can see that the total samples required by \alg tends to be smaller than those required by \texttt{EPS-AP}, \texttt{EXP-GREEDY} and \texttt{EXP-GREEDY-K}, which demonstrates the advantage of \alg in sample efficiency, which was the main goal of the paper. However, on the delicious\_300 dataset (Figures \ref{fig:cover2500_300_eps_average-q} and \ref{fig:cover_300_k-ave_q}), the average samples of \texttt{EXP-GREEDY-K} is slightly better than \alg, and on the other hand \alg has significantly better average samples compared to \texttt{EXP-GREEDY-K} on the corel\_60 dataset (Figures \ref{fig:corel_60_eps_ave_q} and \ref{fig:corel_60_k_ave_q}). This demonstrates the incomparability of the instance-dependent sample query bounds given for marginal gain computations on \alg vs that of \singla.

From the results where we vary $\epsilon$, it can be seen that both the total samples and average samples of our algorithm \alg increase less compared with \texttt{EPS-AP} and \texttt{EXP-GREEDY} as $\epsilon$ decreases (Figures \ref{fig:cover2500_300_eps_q}, \ref{fig:cover2500_300_eps_average-q}, \ref{fig:corel_60_eps_q} and \ref{fig:corel_60_eps_ave_q}), which corresponds to our theoretical results (see the discussion in Section \ref{results} in the appendix).

For the experiments comparing different 
$\kappa$, we can see that the total queries of the \texttt{EXP-GREEDY} and \texttt{EXP-GREEDY-K} increases faster compared with \texttt{EPS-AP} and \alg (Figure \ref{fig:cover_300_k_q}), which can be attributed to the better dependence on $\kappa$ that \threshold exhibits compared to the standard greedy algorithm.
A result that is a little different from the above is that the number of total queries of \texttt{EXP-GREEDY-K} decreases on dataset corel\_60 when $\kappa$ becomes large (Figure \ref{fig:corel_60_k-q}), which is because when $\kappa$ increases, \texttt{EXP-GREEDY-K} is able to better deal with tiny differences in marginal gains (see the appendix).

Finally, the results on the larger dataset (corel and delicious) of \alg and \texttt{EPS-AP} are presented in Figures \ref{fig:cover_eps_q}, \ref{fig:cover_k_q}, \ref{fig:corel_eps-q} and \ref{fig:corel_eps_ave_q}. Notably, our proposed algorithm (\alg) showcases considerable advantages over the \texttt{EPS-AP} algorithm in terms of both required total samples and average samples.



\clearpage
\bibliographystyle{icml2024}
\bibliography{arxiv}

\clearpage

\appendix
\appendixtitle{Appendix}
\section{Additional Related Work}
\label{appdx:related_work}
Approximation algorithms for submodular maximization problems with exact value oracle have been extensively studied in the literature \cite{nemhauser1978analysis,badanidiyuru2014fast,mirzasoleiman2015lazier,balkanski2019exponential}. For MSMC, the standard greedy algorithm produces a solution set with the best possible $1-1/e$ approximation guarantee in $O(n^2)$ queries of $f$. \cite{badanidiyuru2014fast} proposed a faster greedy-like algorithm that gives an approximation guarantee of $1-1/e-O(\epsilon)$ while reducing the sample complexity to $O(\frac{n}{\epsilon}\log\frac{n}{\epsilon})$.

Another variant is USM
\cite{buchbinder2015tight,feige2011maximizing,buchbinder2018deterministic}. Notably, \cite{buchbinder2015tight} introduced a deterministic algorithm that gives a $1/3$ guarantee in $O(n)$ queries to an oracle for $f$, and a randomized version of their algorithm yields the best possible $1/2$ guarantee in expectation in the same number of queries. 

The final variant of submodular maximization we consider is MSMM
\cite{balkanski2019optimal,friedrich2014maximizing,fisher1978analysis}. The greedy algorithm only yields an approximation ratio of $1/2$ in this setting \cite{fisher1978analysis}. But by extending the discrete submodular function to its continuous counterpart, known as the multilinear extension (see the definition in Section \ref{sec:prelim}), and by solving the problem in this regime, it is proved that an approximation ratio arbitrarily close to the best possible $1-1/e$ can be achieved \cite{badanidiyuru2014fast,calinescu2011maximizing}.

Our work is also related to the best-arm-identification in multi-armed bandit literature \cite{audibert2010best,kaufmann2016complexity,jun2016top}, where the objective is to estimate the best action by choosing arms and receiving stochastic rewards from the environment. The most widely considered setting is the PAC learning setting \cite{even2002pac,kalyanakrishnan2012pac,zhou2014optimal}.

Our paper studies the same noisy setting as \cite{singla2016noisy}. There are essentially two versions of \singla, one gives an approximation guarantee of about $1-1/e$ with high probability (like our algorithm \alg does), and the other gives the same approximation guarantee but is randomized. The benefit of the latter over the former is better sample complexity. The bounds given on the sample complexity of \singla and the ones given in this paper for \alg are instance-dependent and incomparable to one another. We discuss how our algorithm relates to \singla in more depth in Section \ref{app:related}, but we briefly list here the potential advantages of our algorithm \alg compared to \singla: (i) Our algorithm has an approximation guarantee of about $1-1/e$ with high probability as opposed to an approximation guarantee of about $1-1/e$ in expectation as in the randomized version of \singla; (ii) Our algorithm is not as sensitive to small differences in marginal gain between elements since it is not based on the standard greedy algorithm as \singla is; (iii) The algorithm of \singla has greater time complexity beyond just the sample complexity because it requires $O(n\log n)$ computations per each noisy query to $\Delta f$; (iv) Our algorithm makes less estimations of $\Delta f$ overall since it is based on a faster variant of the greedy algorithm (\threshold). We further compare the algorithms experimentally in Section \ref{sec:exp_results}.

\subsection{Other Noisy Model}

If the noisy model is that the the samples are taken from distribution $\mathcal{D}(X)$ to evaluate $f(X)$ instead of the marginal gain, the model also satisfies our setting. This is because if the noisy evaluation of $f(X)$ is R-sub-Gaussian, the noisy evaluation of the marginal gain $\Delta f(X,u)$ can be obtained by taking two noisy samples of $f$ and calculating $\mathcal{D}(X\cup\{u\})-\mathcal{D}(X)$ and that the difference of two independent sub-Gaussian random variables is also sub-Gaussian.

\section{Comparison with \singla}
\label{app:related}
In this section, we provide more discussion about the related algorithm \singla of \citet{singla2016noisy}. \singla combines the standard greedy algorithm with the best arm identification algorithm used in combinatorial bandit literature \cite{chen2014combinatorial}.

In particular, the standard greedy algorithm for \prob \cite{nemhauser1978analysis} goes as follows: A solution $S$ is built by iteratively choosing the element $u\in U$ that maximizes the marginal gain $\Delta f(S,u)$ until the cardinality constraint $\kappa$ is exhausted. \singla follows a setting like ours, so instead of choosing the element of maximum marginal gain at each iteration, they follow the standard greedy algorithm but adaptive sampling following techniques from the best-arm identification problem is done in order to identify the element(s) with the highest marginal gain. The simplest version of their algorithm identifies one element with the highest marginal gain at each iteration, and this version has a guarantee of about $1-1/e$ with high probability as in \alg. This algorithm is \texttt{EXP-GREEDY} in Section \ref{sec:exp}. However, a downside of this approach is that many samples are often needed to distinguish between elements of nearly the same marginal gain. In contrast, notice that our algorithm \alg does not need to compare marginal gains between elements and therefore does not have this issue.

In order to deal with the sample inefficiency, \singla is generalized to a randomized version. The randomized version of \singla involves a subroutine called \texttt{TOPX}, which adaptively samples marginal gains until a subset of elements with relatively high marginal gains have been identified. Then a randomly selected element among the subset is added to the solution set. In particular, given an integer $0 <\kappa'\leq \kappa$, the TOPX algorithm runs TOP-$l$ selection algorithms for each $l\in\{1,2,...,\kappa'\}$, and each of the TOP-$l$ selection algorithm runs until it returns a subset of $l$ items with highest marginal gain with high probability. The TOPX algorithm stops once there exists some $l$ such that the TOP-$l$ selection algorithm ends. This randomized version of \singla has an almost $1-1/e$ approximation guarantee, but it holds in expectation and with high probability. The case where $\kappa'=\kappa$ is \texttt{EXP-GREEDY-K} in Section \ref{sec:exp}.

Now that we have described the two versions of \singla and their corresponding approximation guarantee, we look into more detail about the efficiency of \singla in terms of runtime and sample complexity.

It is proven by \cite{singla2016noisy} that the number of samples taken for each iteration where an element is added to the solution is at most
\begin{align*}
    O\left(n\kappa'R^2\min\left\{\frac{4}{\Delta_{\max}^2},\frac{1}{\epsilon^2}\right\}\log\left(\frac{R^2\kappa n\min\{\frac{4}{\Delta_{\max}^2},\frac{1}{\epsilon^2}\}}{\delta}\right)\right)
\end{align*}
where $\Delta_{\max}$ is the largest difference amongst the first $\kappa'$ element's marginal gains. In other words, this is the number of samples taken each time TOPX is called. Since an element being added involves approximating the marginal gains over all of the elements of $U$, the average sample complexity to compute an approximate marginal gain for a single element is then
\begin{align*}
    O\left(\kappa'R^2\min\left\{\frac{4}{\Delta_{\max}^2},\frac{1}{\epsilon^2}\right\}\log\left(\frac{R^2\kappa n\min\{\frac{4}{\Delta_{\max}^2},\frac{1}{\epsilon^2}\}}{\delta}\right)\right).
\end{align*}
We compare the above to a single call of \samp in our algorithm \alg, which is the analogous computation where we are approximating the marginal gain for an element of $U$. Recall from Theorem \ref{mainthm} that the bound for the sample complexity for \samp is the minimum between
    \begin{align*}
        \left\{\frac{2R^2}{\phi^2(S,u)}\log\left(\frac{4R^2\sqrt{\frac{3nh(\alpha)}{\delta}}}{\phi^2(S,u)}\right),\frac{R^2}{2\epsilon^2}\log \left(\frac{6nh(\alpha)}{\delta}\right)\right\}.
    \end{align*}
If $k'=1$, i.e. the non-randomized version of \singla that has a similar approximation guarantee to our algorithm \alg, then $\Delta_{\text{max}}$ is the difference between the top two marginal gains, which could be very small and therefore the sample complexity quite high. On the other hand, \samp is not sensitive to this property. In order to make $\Delta_{\max}$ bigger, one could increase $k'$ and use the randomized version of \singla. But this case could have worse sample complexity compared to ours as well. If $\Delta_{\max}$ is small and satisfies that $\Delta_{\max}=O(\epsilon)$, then the sample complexity of \singla is worse than our averaged sample complexity by a factor of at least $O(\kappa')$.

Further, since \singla follows the standard greedy algorithm, there are $\kappa$ calls made to TOPX. In contrast, \alg is based on the faster variant of the greedy algorithm, \threshold, and so only requires $O(\log(\kappa))$ iterations over $U$.

Another factor that makes \alg preferable to \singla is its run time besides sample complexity.  From the description of \singla in \cite{singla2016noisy}, we can see that at each time a noisy query to $\Delta f$ is taken, the TOP-$l$ selection algorithm updates the confidence interval for all the elements, and then the algorithm sorts all elements to find the set $M_t$ of $l$ elements with highest empirical marginal gain. Then another estimate of the marginal gains is computed to be the empirical mean plus a confidence interval or minus the confidence interval depending on whether the elements are within $M_t$. Next, the algorithm sorts the newly obtained estimates to find the top-$l$ set with respect to the new estimates. However, both \alg and \texttt{EPS-AP} have more efficient runtime complexity and require only one update of the confidence interval in Line \ref{alg: update confidence interval} and two comparisons in Line \ref{line: comparison to thres 1} and \ref{line: comparison to thres 2} in \samp, which is only $O(1)$ in computation.

\section{Appendix for Section \ref{sec:sampling}}

In this section, we present the omitted content of Section \ref{sec:sampling}. In Section \ref{appdx:proof_of_samp}, we present the proof of Theorem \ref{thm:sampling}. In Section \ref{appdx:proof_of_samp2}, we present the proof of Theorem \ref{thm:sampling2}.

\subsection{Additional Lemmas and Analysis of Theorem \ref{thm:sampling}}

\label{appdx:proof_of_samp}

In this section, we present the proof of Theorem \ref{thm:sampling2}, which provides the theoretical results of sample complexity and approximation guarantee of the \samp algorithm. First of all, we provide the statement of Theorem \ref{thm:sampling} again.

\noindent\textbf{Theorem \ref{thm:sampling}. }\textit{
   For any random variable $X$ that is $R$-sub-Gaussian, if we define $N_1=2R^2/\epsilon^2\log \frac{4}{\delta}$, and $\conf =R\sqrt{\frac{2}{t}\log \frac{8 t^2}{\delta}}$, then the algorithm \samplong achieves that with probability at least $1-\delta$
    \begin{enumerate}
        \item \samp on input $(w,\epsilon,\delta,\mathcal{D}_X,R)$ takes at most the minimum between
    \begin{align*}
        \left\{\frac{2R^2}{\epsilon^2}\log \left(\frac{4}{\delta}\right),\frac{8R^2}{\phi_X^2}\log\left(\frac{16R^2}{\phi_X^2}\sqrt{\frac{2}{\delta}}\right)\right\}
    \end{align*}
    noisy samples, where $R$ is as defined in Section \ref{sec:prelim}, $\phi_X = \frac{\epsilon + |w-\mathbb{E} X|}{2}$.
    \item If \samp returns true, then $\mE X\geq w-\epsilon$. If \samp returns false, then $\mE X\leq w+\epsilon$.
    \end{enumerate}
    }

Before we present the detailed proof, here we provide an overview of the proof. In order for \samp to correctly determine whether $\mE X$ is approximately above or below the threshold $w$, i.e. the second result of Theorem \ref{thm:sampling}, two random events must occur during \samp. The first event is that at all iterations during the for loop, the confidence regions around the sample mean ($\hat{X}_t$) contain the true expected value ($\mE X$). The second event is that after $N_1$ samples taken by the for loop on Line \ref{line:sample_N_1}, we have achieved an $\epsilon$-additive approximation of the expected value. Basically these two events together mean that \samp is correct about the region where $\mE X$ is throughout the algorithm, and therefore it returns the correct answer to whether $\mE X$ is approximately above or below the threshold $w$. The following Lemma states that on a run of \samp, the two events hold with probability at least $1-\delta$.

\begin{lemma}
\label{lem:clean_event}
    With probability at least $1-\delta$, the following two events hold.
    \begin{enumerate}       
    \item At any time $t\in\mathbb{N}_+$, the sample mean $\widehat{X}_t$ satisfies that
    $|\widehat{X}_t-\mE X|\leq \conf$,
    where $\conf:=R\sqrt{\frac{2}{t}\log \frac{8 t^2}{\delta}}$.
    \item The sample mean $\widehat{X}_{N_1}$ at time $N_1:=\frac{2R^2}{\epsilon^2}\log \frac{4}{\delta}$ satisfies that $|\widehat{X}_{N_1}-\mE X|\leq \epsilon  $.
    \end{enumerate}
\end{lemma}
\begin{proof}
    First, we apply the Hoeffding's inequality on $\widehat{X}_{N_1}$ and it follows that 
    \begin{align*}
        P\left(|\widehat{X}_{N_1}-\mE X|\geq \epsilon \right)\leq 2\exp\left(-\frac{N_1\epsilon^2}{2R^2}\right)\leq\frac{\delta}{2}.
    \end{align*}
    Next, by applying the Hoeffding's inequality for any fixed time $t$, we have that
    \begin{align*}
        P\left(|\widehat{X}_t-\mE X|\geq \conf \right)\leq \frac{\delta}{4t^2}.
    \end{align*}
By taking the union bound for any time $t$, it follows that
\begin{align*}
    &P(\exists t \text{ s.t. }|\widehat{X}_t-\mE X|\geq \conf )\\
    &\leq\sum_{t=1}^\infty P(|\widehat{X}_t-\mE X|\geq \conf )\\
    &\leq \frac{\delta}{4}\sum_{t=1}^\infty
    \frac{1}{t^2}\leq\frac{\delta}{2}.
\end{align*}
By taking the union bound again on the two events above, we have that
\begin{align*}
    &P(|\widehat{X}_{N_1}-\mE X|\geq \epsilon\text{ or }\exists t \text{ s.t. }|\widehat{X}_t-\mE X|\geq \conf )\\
    &\leq P\left(|\widehat{X}_{N_1}-\mE X|\geq \epsilon \right)+P(\exists t \text{ s.t. }|\widehat{X}_t-\mE X|\geq \conf )\\
    &\leq \delta.
\end{align*}
\end{proof}

The second lemma required for establishing Theorem \ref{thm:sampling} concerns the number of samples that \samp takes before its approximation of $\mE X$ is sufficiently accurate so that it can terminate. The number of samples depends on how far away the true value of $f$ is from the threshold. In particular, Lemma \ref{lem:conf_int} below states that once the confidence interval goes beneath the corresponding $\phi$ value (as defined in Theorem \ref{thm:sampling}), then \samp will complete.
Lemma \ref{lem:conf_int} and its proof are stated below.
\begin{lemma}
\label{lem:conf_int}
    With probability at least $1-\delta$, when the confidence interval $\conf$ satisfies that
    \begin{align*}
        \conf \leq \phi_X,
    \end{align*}
    the sampling of $ X$ finishes, where $\phi_X = \frac{\epsilon + |w-\mE X|}{2}$.
    
\end{lemma}
\begin{proof}
     If $\conf\leq\frac{\epsilon+w-\mE X}{2}$, then we have $\mE X\leq w+\epsilon-2\conf $. From Lemma \ref{lem:clean_event}, we have that with probability at least $1-\delta$, it holds that $\widehat{X}_t-\mE X\leq\conf$. Therefore,
     \begin{align*}
         &\widehat{X}_t+\conf\\
         &\leq (\widehat{X}_t-\mE X)+\mE X+\conf\\
         &\leq w+\epsilon.
     \end{align*}
     Thus the algorithm ends.
     
     Similarly, we consider the case where $\conf\leq\frac{\epsilon-w+\mE X}{2}$. In this case, we have that $\mE X\geq 2\conf+w-\epsilon$.
    Notice that conditioned on the clean event defined in Lemma \ref{lem:clean_event}, we have that $\widehat{X}_t-\mE X\geq-\conf$. Then
    \begin{align*}
        \widehat{X}_t-\conf&\geq \widehat{X}_t-\mE X\\
        &\qquad+\mE X-\conf\\
        &\geq-\conf+2\conf\\
        &\qquad+w-\epsilon-\conf\\
        &=w-\epsilon.
    \end{align*}
    Therefore, the algorithm ends. 
\end{proof}

Now we present the proof of Theorem \ref{thm:sampling}. 
\begin{proof}
    We first prove the result on sample complexity, which is the first result in Theorem \ref{thm:sampling}. From Lemma \ref{lem:conf_int}, we have if
    \begin{align}
    \label{ineq:confidence_intv}
     \conf \leq \phi_X,
    \end{align}
    then the Algorithm \ref{alg:samp} finishes.
    Since $\conf =R\sqrt{\frac{2}{t}\log \frac{8 t^2}{\delta}}$, we have the above inequality (\ref{ineq:confidence_intv}) is equivalent to that
    \begin{align*}
     \frac{4\log (\sqrt{\frac{8}{\delta}}t)}{t}\leq \frac{\phi^2_X}{R^2}.
    \end{align*}
    Since $\sqrt{\frac{8}{\delta}}t\geq 2$, from Lemma \ref{lem:logx_over_x}, we have when
    \begin{align*}
        t\geq\frac{8R^2}{\phi^2_X}\log(\frac{16R^2}{\phi^2_X}\sqrt{\frac{2}{\delta}}),
    \end{align*}
    the above inequality holds and the Algorithm \ref{alg:samp} ends. Therefore, the number of samples required is bounded by $\min\{\frac{8R^2}{\phi^2_X}(\log\frac{16R^2}{\phi^2_X}\sqrt{\frac{2}{\delta}}),N_1\}$. 

    Next, we prove the second result in Theorem \ref{thm:sampling}. If $t=N_1$ when \samp ends, then conditioned on the events in Lemma \ref{lem:clean_event}, $|\widehat{X}_{N_1}-\mE X|\leq\epsilon$. Thus 
  if the algorithm returns true, $\mE X\geq \widehat{X}_t - \epsilon\geq w-\epsilon$. If the output of the algorithm is false, then $\widehat{X}_t \leq w$. Similarly we have that $\mE X\leq \widehat{X}_t + \epsilon\leq w+\epsilon$. Secondly, let us consider the case where $t<N_1$ when the algorithm \samp ends. Conditioned on the second event in Lemma \ref{lem:clean_event}, we have if the algorithm \samp returns true, $\mE X\geq\widehat{X}_t-\conf\geq w-\epsilon$. If the output is false, $\mE X\leq\widehat{X}_t+\conf\leq w+\epsilon$.
\end{proof}

\subsection{Proof and Analysis of Theorem \ref{thm:sampling2}}
\label{appdx:proof_of_samp2}
In this section, we present the omitted proofs of Theorem \ref{thm:sampling2} in Section \ref{sec:sampling}. Theorem \ref{thm:sampling2} provides another result of the approximation error for the \samp algorithm by defining the confidence interval $C_t$ to be $C_t=\frac{3R}{t\alpha}\log\big(\frac{8t^2}{\delta}\big)$ and the worst-case sample complexity $N_1$ to be $N_1=\frac{3R}{\epsilon\alpha}\log \left(\frac{4}{\delta}\right)$. We begin by stating Theorem \ref{thm:sampling2}, followed by the proof of the theorem. Finally, we establish the lemmas crucial to the proof of the theorem.

\noindent\textbf{Theorem \ref{thm:sampling2}. }\textit{
   For any random variable $X$ that is bounded in the range of $[0,R]$, if we define $C_t=\frac{3R}{t\alpha}\log(\frac{8t^2}{\delta})$, and $N_1=\frac{3R}{\epsilon\alpha}\log \left(\frac{4}{\delta}\right)$ where $\alpha$ is an additional parameter that controls the multiplicative error rate, the algorithm \samplong achieves that with probability at least $1-\delta$, the algorithm \sampnewlong achieves that with probability at least $1-\delta$
    \begin{enumerate}
        \item \sampnew on input $(w,\epsilon,\delta,\mathcal{D}_X,R)$ takes at most the minimum between
    \begin{align*}
        \left\{\frac{3R}{\epsilon\alpha}\log \left(\frac{4}{\delta}\right),\frac{12R}{\alpha\phi_X'}\log\left(\frac{12R}{\alpha\phi_X'}\sqrt{\frac{8}{\delta}}\right)\right\}
    \end{align*}
    noisy samples, $\phi_X' = \frac{\epsilon -\alpha\mE X+| w-\mathbb{E} X|}{2}$.
    \item If the output is true, then $(1+\alpha)\mE X\geq w-\epsilon$. If the output is false, then $(1-\alpha)\mE X\leq w+\epsilon$.
    \end{enumerate}
}
\begin{proof}
   First of all, we prove the result on the sample complexity as presented in the first result in Theorem \ref{thm:sampling2}. From Lemma \ref{lem:conf_int2}, we have if $$\conf\leq\phi_X',$$ the algorithm ends. By definition of $\conf$, we have that the above result is equivalent to that
   \begin{align*}
       \frac{3R}{t\alpha}\log(\frac{8t^2}{\delta})\leq\phi_X'.
   \end{align*}
   From Lemma \ref{lem:logx_over_x}, we have that when
   \begin{align*}
       t\geq\frac{12R}{\alpha\phi_X'}\log\big(\frac{12R}{\alpha\phi_X'}\sqrt{\frac{8}{\delta}}\big)
   \end{align*}
   the above inequality holds and thus the algorithm ends. From the description of the algorithm, we have that the number of samples is also bounded by $N_1$. Therefore, the first result in Theorem \ref{thm:sampling2} is proved.

   Next, we prove the second result on the difference of $\mE X$ and $w$. If $t=N_1$ when \samp ends, then if the algorithm returns true, we have that with probability at least $1-\delta$,
   \begin{align*}
       (1+\alpha)\mathbb{E}X+\epsilon\geq \widehat{X}_{N_1}\geq w.
   \end{align*}
   where the first inequality follows from Lemma \ref{lem:clean_event_samp2}. If the algorithm returns false and $t=N_1$ when the algorithm ends, then with probability at least $1-\delta$, 
   \begin{align*}
       (1-\alpha)\mathbb{E}X-\epsilon\leq \widehat{X}_{N_1}\leq w.
   \end{align*} 
  Next, we consider the case where $t<N_1$ when the algorithm ends. Conditioned on the first event in Lemma \ref{lem:clean_event_samp2} and from the stopping condition of \samp, we can see if \samp returns true, then
  \begin{align*}
       (1+\alpha)\mathbb{E}X+\epsilon\geq \widehat{X}_t-\conf+\epsilon\geq w.
   \end{align*}
   If \sampnew returns false, then
  \begin{align*}
       (1-\alpha)\mathbb{E}X-\epsilon\leq \widehat{X}_t+\conf-\epsilon\leq w.
   \end{align*} 
\end{proof}

We now present the statement and the proofs of the lemmas used in the proof of Theorem \ref{thm:sampling2}. We start by introducing Lemma \ref{lem:clean_event_samp2}, which defines 
two "clean events".
\begin{lemma}
    \label{lem:clean_event_samp2}
    With probability at least $1-\delta$, the following two events hold.
    \begin{enumerate}       
    \item At any time $t\in\mathbb{N}_+$, the sample average $\widehat{X}_t$ satisfies that
    $|\widehat{X}_t-\mE X|\leq \alpha\mathbb{E}X+\conf$,
    where $\conf:=\frac{3R}{t\alpha}\log(\frac{8t^2}{\delta})$.
    \item The sample average $\widehat{X}_{N_1}$ at time $N_1:=\frac{3R}{\epsilon\alpha}\log \left(\frac{4}{\delta}\right)$ satisfies that $|\widehat{X}_{N_1}-\mE X|\leq \alpha\mE X+\epsilon$.
    \end{enumerate}
\end{lemma}
\begin{proof}
    By applying the Lemma \ref{lem:chernoff}, we have that for any fixed time step $t$, 
    \begin{align*}
        P\big(|\widehat{X}_t-\mE X|> \alpha\mathbb{E}X+\conf\big)&\leq2\exp\{-\frac{t\alpha\conf}{3R}\}\\
        &\leq\frac{\delta}{4t^2}.
    \end{align*}
    By taking the union bound over all time step $t\in\mathbb{N}_+$, we have
    \begin{align*}
        &P\big(|\widehat{X}_t-\mE X|> \alpha\mathbb{E}X+\conf,\forall t\big)\\
        \leq&\sum_{t=1}^{\infty}P\big(|\widehat{X}_t-\mE X|> \alpha\mathbb{E}X+\conf\big)\\
    \leq&\sum_{t=1}^{\infty}\frac{\delta}{4t^2}\leq\frac{\delta}{2}.
    \end{align*}
    Therefore the first event in the lemma holds with probability at least $1-\delta/2$. By applying the Lemma \ref{lem:chernoff} again, we have that for $t=N_1$, 
    \begin{align*}
        P\big(|\widehat{X}_{N_1}-\mE X|> \alpha\mathbb{E}X+\epsilon\big)\leq2\exp\{-\frac{N_1\alpha\epsilon}{3R}\}=\delta/2.
    \end{align*}
     It follows that the second event in the lemma holds with probability at least $1-\delta/2$. By combining the two results and applying the union bound again, we know that with probability at least $1-\delta$, the two events both hold.
\end{proof}
Next, we prove another lemma that is used in the proof of the sample complexity result in Theorem \ref{thm:sampling2}.

\begin{lemma}
\label{lem:conf_int2}
    With probability at least $1-\delta$, when the confidence interval $\conf$ satisfies that
    \begin{align*}
        \conf \leq \phi'_X,
    \end{align*}
    the sampling of $ X$ finishes, where $\phi_X' = \frac{\epsilon-\alpha\mE X + |w-\mE X|}{2}$.
\end{lemma}
\begin{proof}
    To prove the lemma, it is equivalent to prove that when $C_t\leq\frac{\epsilon-\alpha\mE X + w-\mE X}{2}$ or $C_t\leq\frac{\epsilon-\alpha\mE X - w+\mE X}{2}$, the algorithm ends. First of all, if $C_t\leq\frac{\epsilon-\alpha\mE X + w-\mE X}{2}$, then $(1+\alpha)\mE X+2\conf\leq w+\epsilon$. Conditioned on the events in Lemma \ref{lem:clean_event_samp2}, we have that with probability at least $1-\delta$, it follows that
    \begin{align*}
        \widehat{X}_t+\conf\leq(1+\alpha)\mE X+2\conf\leq w+\epsilon.
    \end{align*}
    Thus the sampling of $X$ ends. Next, if $C_t\leq\frac{\epsilon-\alpha\mE X - w+\mE X}{2}$, then $(1-\alpha)\mE X-2\conf\geq w-\epsilon$. By Lemma \ref{lem:clean_event_samp2},
    \begin{align*}
        \widehat{X}_t-\conf\geq(1-\alpha)\mE X-2\conf\geq w-\epsilon.
    \end{align*}
    Then the algorithm ends.
\end{proof}
\section{Appendix for Section \ref{sec:monotone}}
In this section, we present the omitted content in Section \ref{sec:monotone}, which is organized as follows: In Section \ref{appdx:compare_to_sample_before}, we discuss and compare the theoretical performance of our algorithm, \algmono, with the sampling-before-hand algorithm in the context of the influence maximization problem. Next, we provide the proof of our main result, Theorem \ref{mainthm}, in Section \ref{appdx:proof_of_mono}. Theorem \ref{mainthm} gives the theoretical guarantee of the \alg algorithm. Finally, in Section \ref{appdx:proof_of_mono2}, we provide the brief description of \algmono algorithm and the detailed proof of Theorem \ref{thm:monotone2}.

\subsection{Comparing to sampling-before-hand algorithm}
\label{appdx:compare_to_sample_before}
Before we describe the sampling-before-hand algorithm and dive into the comparison of this algorithm and \algmono, first we present a detailed description of the application of influence maximization. In the influence maximization problem in large-scale networks, the submodular objective is defined as follows:

{\textbf{Influence Maximization }} Suppose the social graph is described by $G=(V,E,\bar{\vect{w}})$, where $V$ is the set of nodes with $|V|=n$, $E$ denotes the set of edges, and $\bar{\vect{w}}$ is the weight vector defined on the set of edges $E$. Given a seed set $S$, let us define $f(S;\vect{w})$ to be the number of nodes reachable from the seed set $S$ under the graph realizations determined by a random weight vector $\vect{w}$. Therefore, $f(S;\vect{w})$ is bounded by the number of nodes in the graph, i.e., $0\leq f(S;\vect{w})\leq n$. The submodular objective is defined as $f(S)=\mE_{\vect{w}\sim\mathcal{D}(\bar{\vect{w}})}f(S;\vect{w})$. Here $\mathcal{D}(\bar{\vect{w}})$ is the distribution of the weight vector.  

The marginal gain can be calculated as 
\begin{align*}
\Delta f(S,s)&=\mE_{\vect{w}\sim\mathcal{D}(\bar{\vect{w}})}\Delta f(S,s;\vect{w})\\
&=\mE_{\vect{w}\sim\mathcal{D}(\bar{\vect{w}})}f(S\cup\{u\};\vect{w})-\mE_{\vect{w}\sim\mathcal{D}(\bar{\vect{w}})} f(S),
\end{align*}
which is also bounded in the range of $[0,n]$.

Next, we describe the sampling-before-hand algorithm, which runs as follows:
\begin{enumerate}
    \item \textbf{Sampling:} The algorithm begins by sampling $N$ i.i.d graph realizations. For the $i$-th graph realization, we denote its weight vector as $\vect{w}_i$ and the corresponding function value for a set $S$ as $f_i(S)=f(S;\vect{w}_i)$.
    \item \textbf{Average Objective Function:} Next, we define the average function $\hat{f}$ over the sampled graph realizations. This function is given by $\hat{f}(S)=\frac{\sum_{i=1}^Nf_i(S)}{N}$ for any $S\subseteq U$.
    \item \textbf{Threshold-Greedy Algorithm: }We run \thresholdlong (\threshold) with the average function $\hat{f}$ as the submodular objective. The output of the threshold-greedy algorithm is returned as the solution set, denoted as $S$.
\end{enumerate}

\subsubsection{Analysis of Sampling-before-hand approach}
Now we present the analysis of the sampling-before-hand algorithm. From Lemma \ref{lem:chernoff}, and by taking the union bound, we can prove that 
\begin{align*}
    P(|\hat{f}(X)&-f(X)|\geq\alpha f(X)+\epsilon, \forall |X|\leq \kappa)\\
    &\leq2n^{\kappa}\exp\{-\frac{N\alpha\epsilon}{3n}\}.
\end{align*}
Therefore, to guarantee that $$P(|\hat{f}(X)-f(X)|\geq\alpha f(X)+\epsilon, \forall |X|\leq \kappa)\leq\delta,$$ it is enough to take 
$$N\in\Omega\big(\frac{n}{\alpha\epsilon}(\kappa\log n+\log\frac{1}{\delta})\big)$$ number of graph realizations. Since \threshold requires $\frac{n}{\alpha}\log\frac{n}{\alpha}$ number of evaluations of $\hat{f}$. The total number of evaluations of noisy realizations of $f$ would be 
\begin{align*}
    O\big(\frac{n^2}{\alpha^2\epsilon}\log\frac{n}{\alpha}(\kappa\log n+\log\frac{1}{\delta})\big).
\end{align*}
Next, we prove the approximation guarantee. From the analysis above, we can see that with probability at least $1-\delta$
\begin{align*}
   f(S)&\geq\frac{\hat{f}(S)-\epsilon}{1+\alpha} \\
   &\geq(1-\alpha)\hat{f}(S)-\epsilon \\
   &\geq (1-1/e-\alpha)(1-\alpha)\hat{f}(OPT)-\epsilon\\
   &\geq (1-1/e-2\alpha)\hat{f}(OPT)-\epsilon\\
   &\geq (1-1/e-3\alpha)f(OPT)-2\epsilon.
\end{align*}
Now we compare the theoretical guarantees of the sampling-based algorithm and \algmono. The theoretical results of \algmono are in Theorem \ref{thm:monotone2}. Notice that by substituting $\epsilon$ with $\epsilon/k$ in Theorem \ref{thm:monotone2}, we obtain a similar approximation guarantee for \algmono: $f(S)\geq (1-1/e-O(\alpha))f(OPT)-O(\epsilon)$, which matches the result achieved by the sampling-based algorithm. 

For the sample complexity, each call of \sampnew requires at most the minimum between
$O(\frac{\kappa n}{\epsilon\alpha}\log\frac{n}{\delta})$ and $O(\frac{n}{\alpha\phi'(S,u)}\log\frac{n}{\alpha\phi'(S,u)\delta})$ number of samples. The first bound is derived by considering the fixed $\epsilon$- approximation of the marginal gain. If we only consider this bound, then the total number of marginal gains would be $O(\frac{kn^2}{\epsilon\alpha^2}(\log\frac{n}{\alpha})(\log\frac{n}{\delta}))$. In practice, the parameter $\delta$ is usually set to be $O(Poly(1/n))$, such as $O(1/n^2)$. Consequently, the sample complexity of both \algmono and the sampling-before-hand approach would be $O(\frac{\kappa n}{\epsilon\alpha}\log n)$. However, it is important to note that \sampnew employs the adaptive thresholding technique, which often allows the algorithm to terminate much earlier before reaching the worst-case sample complexity required for fixed-confidence approximation. As a result, \algmono can be significantly more sample-efficient in practice.

In comparison to the sampling-before-hand algorithm, \algmono offers an additional advantage. The sampling-before-hand algorithm requires obtaining $N$ independent graph realizations and storing all the data at the beginning of the algorithm. However, this can pose practical challenges. Firstly, in scenarios where both $N$ and the graph are exceedingly large, storing all the data might be infeasible. Secondly, in certain applications, such as real-world social networks, obtaining an entire graph realization may not be possible, as we might only be able to sample a portion of the graph at each time.

\subsection{Proof of Theorem \ref{mainthm}}
\label{appdx:proof_of_mono}
In this section, we move towards proving one of our main results, Theorem \ref{mainthm} about \alg for the MSMC problem. We state the theorem again as follows.

\noindent\textbf{Theorem \ref{mainthm}. }\textit{
    Suppose the noisy marginal gain of any subset $S\subseteq U$ and element $s\in U$ is $R$-sub-Gaussian, then \alg makes at most $n\log(\kappa/\alpha)/\alpha$ calls of \samp. In addition, with probability at least $1-\delta$, the following statements hold:
    \begin{itemize}
        \item The exact function value of the output solution set $S$ satisfies that $f(S)\geq(1-e^{-1}-\alpha)f(OPT)-2\kappa\epsilon$;
    \item Each call of \samp on input ($w$, $\epsilon$, $\frac{2\delta}{3nh(\alpha)}$, $\mathcal{D}(S,u)$, $R$) takes at most the minimum between
    \begin{align*}
        \frac{8R^2}{\phi^2(S,u)}\log\left(\frac{16R^2\sqrt{\frac{3nh(\alpha)}{\delta}}}{\phi^2(S,u)}\right)
    \end{align*}
    and
    \begin{align*}
       \frac{2R^2}{\epsilon^2}\log \left(\frac{6nh(\alpha)}{\delta}\right)
    \end{align*}
    noisy samples. Here $OPT$ is an optimal solution to the MSMC problem, $\phi(S,u) = \frac{\epsilon + |w-\Delta f(S,u)|}{2}$, and $h(\alpha)=\frac{\log{(\kappa/\alpha)}}{\alpha}$.
    \end{itemize}
}


To prove the theorem, we first present a series of needed lemmas. In order for the guarantees of Theorem \ref{mainthm} to hold, two random events must occur during \alg. The first event is that the estimate of the max singleton value of $f$ on Line \ref{alg:ATG:line:sample-mean} in \alg is an $\epsilon$-approximation of its true value. More formally, we have the following lemma.
 \begin{lemma}
    \label{lem:clean_event_monotone}
      With probability at least $1-\delta/3$, we have $\max_{s\in U}f(s)-\epsilon\leq d\leq\max_{s\in U}f(s)+\epsilon$.
\end{lemma}
\begin{proof}
    For a fix $s\in U$, by Hoeffding's inequality we would have that
    \begin{align}
        P(|\hat{f}(s)-f(s)|\geq\epsilon)\leq\frac{\delta}{3n}.
    \end{align}
    Taking a union bound over all elements we would have that
    \begin{align*}
        P(\exists s\in U, s.t.|\hat{f}(s)-f(s)|\geq\epsilon )\leq\frac{\delta}{3}.
    \end{align*}
    Then with probability at least $1-\frac{\delta}{3}$, $|\hat{f}(s)-f(s)|\leq\epsilon$ 
    for all $s\in U$.  
    It then follows that $\forall s\in U$, $f(s)-\epsilon\leq\hat{f}(s)\leq f(s)+\epsilon$. Therefore $$\max_{s\in U}(f(s)-\epsilon)\leq\max_{s\in U}\hat{f}(s)\leq\max_{s\in U}(f(s)+\epsilon).$$ Thus we have
    \begin{align*}
        \max_{s\in U}f(s)-\epsilon\leq d\leq\max_{s\in U}f(s)+\epsilon.
    \end{align*} 
    
\end{proof}

The second event is that for all calls of \samp, the result in Theorem \ref{thm:sampling} holds, which is stated formally as follows. 
 \begin{lemma}
 \label{lem:clean_event_call_to_CS_mono}
     With probability at least $1-2\delta/3$, we have that during each call of \samp with the solution set $S$ and element $u$, the output satisfies that if $thre$ is true, then $\Delta f(S,u)\geq w-\epsilon$. If $thre$ is false, then $\Delta f(S,u)\leq w+\epsilon$.
 \end{lemma}
 \begin{proof}
    First, since each sampling result of the marginal gain is assumed to be $R$-sub-Gaussian, by applying the result in Theorem \ref{thm:sampling}, we can prove that for each call of \samp during \alg with a fixed solution set $S$ and evaluated element $u$ as input, and with probability at least $1-\frac{2\delta}{3nh(\alpha)}$, if the output of \samp is true, then $\Delta f(S,u)\geq w-\epsilon$. Otherwise, $\Delta f(S,u)\leq w+\epsilon$. Since there are $n$ elements in the universe and the number of iterations in Algorithm \ref{alg:ATG} is bounded by $\frac{\log{\kappa/\alpha}}{\log(1/(1-\alpha))}\leq h(\alpha)$, there are at most $nh(\alpha)$ number of marginal gains to evaluate in Algorithm \ref{alg:ATG}. Therefore, by taking the union bound we have that with probability at least $1-2\delta/3$, the statement holds.
\end{proof}

With the above Lemma \ref{lem:clean_event_monotone}  and Lemma \ref{lem:clean_event_call_to_CS_mono}, and by taking the union bound, we have that with probability at least $1-\delta$, the two events both hold during the \alg. Our next step is to show that if both of the events occur during \alg, the approximation guarantees and sample complexity of Theorem \ref{mainthm} hold. To this end, we need the following Lemma \ref{lem:mar_gain}. 

\begin{lemma}
\label{lem:mar_gain}
    Assume the events defined in Lemma \ref{lem:clean_event_monotone} and Lemma \ref{lem:clean_event_call_to_CS_mono} above hold during \alg. Then for any element $s$ that is added to the solution set $S$, the following statement holds.
    \begin{align*}
        \Delta f(S,s)\geq \frac{1-\alpha}{\kappa}(f(OPT)-f(S))-2\epsilon.
    \end{align*}
\end{lemma}

\begin{proof}   
    At the first iteration, if an element $s$ is added to the solution set, it holds by Lemma \ref{lem:clean_event_monotone} that $\Delta f(S,s)\geq w-\epsilon$.
   Since at the first iteration $w=d$ and $d\geq \max_{s\in U}f(s)-\epsilon$. It follows that $\Delta f(S,s)\geq \max_{s\in U}f(s)-2\epsilon$.
    By submodularity we have that $\kappa\max_{s\in U}f(s)\geq f(OPT)$. Therefore, $\Delta f(S,s)\geq \frac{f(OPT)-f(S)}{\kappa}-2\epsilon$.
    
    At iteration $i$ where $i>1$, if an element $o\in OPT$ is not added to the solution set, then it is not added to the solution at the last iteration, where the threshold is $\frac{w}{1-\alpha}$. By Lemma \ref{lem:clean_event}, we have
    $\Delta f(S,o)\leq \frac{w}{1-\alpha}+\epsilon$.
    Since for any element $s$ that is added to the solution at iteration $i$, by Lemma \ref{lem:clean_event} it holds that $\Delta f(S,s)\geq w-\epsilon$. Therefore, we have
    \begin{align*}
        \Delta f(S,s)&\geq w-\epsilon \\
        &\geq(1-\alpha)(\Delta f(S,o)-\epsilon)-\epsilon\\
        &\geq (1-\alpha)\Delta f(S,o)-2\epsilon. 
    \end{align*}
    By submodularity, it holds that $\Delta f(S,s)\geq (1-\alpha)\frac{f(OPT)-f(S)}{\kappa}-2\epsilon$.
\end{proof}
We now prove the main result, Theorem \ref{mainthm}, which relies on the previous Lemma \ref{lem:clean_event_monotone}, \ref{lem:clean_event_call_to_CS_mono} and \ref{lem:mar_gain}.
\begin{proof}
    The events defined in Lemma \ref{lem:clean_event_monotone}, \ref{lem:clean_event_call_to_CS_mono} hold with probability at least $1-\delta$ by combining Lemma \ref{lem:clean_event_monotone}, \ref{lem:clean_event_call_to_CS_mono}, and taking the union bound. Therefore in order to prove Theorem \ref{mainthm}, we assume that both the two events have occurred. 
    The proof of the first result in the theorem depends on the Lemma \ref{lem:mar_gain}. First, consider the case where the output solution set satisfies $|S|=\kappa$. Denote the solution set $S$ after the $i$-th element is added as $S_i$. Then by Lemma \ref{lem:mar_gain}, we have
    \begin{align*}
        f(S_{i+1})\geq \frac{1-\alpha}{\kappa}f(OPT)+(1-\frac{1-\alpha}{\kappa})f(S_i)-2\epsilon.
    \end{align*}
By induction, we have that
\begin{align*}
    f(S_{\kappa})&\geq(1-(1-\frac{1-\alpha}{\kappa})^k)\{f(OPT)-\frac{2\kappa\epsilon}{1-\alpha}\}\\
    &\geq (1-e^{-1+\alpha})\{f(OPT)-\frac{2\kappa\epsilon}{1-\alpha}\}\\
    &\geq(1-e^{-1}-\alpha)\{f(OPT)-\frac{2\kappa\epsilon}{1-\alpha}\}\\
    &\geq(1-e^{-1}-\alpha)f(OPT)-2\kappa\epsilon.
\end{align*}
 If the size of the output solution set $S$ is smaller than $\kappa$, then any element $o\in OPT$ that is not added to $S$ at the last iteration satisfies that
$\Delta f(S,o)\leq w+\epsilon$.
Since the threshold $w$ in the last iteration satisfies that $w\leq\frac{\alpha d}{\kappa}$, we have
\begin{align*}
    \Delta f(S,o)&\leq\frac{\alpha d}{\kappa}+\epsilon.
\end{align*}
It follows that
\begin{align*}
    \sum_{o\in OPT\backslash S}\Delta f(S,o)
    &\leq\alpha (\max_{s\in S}f(s)+\epsilon)+\kappa\epsilon\\
    &\leq\alpha f(OPT)+2\kappa\epsilon.
\end{align*}
 By submodularity and monotonicity of $f$, we have $f(S)\geq (1-\alpha)f(OPT)-2\kappa\epsilon$.

\end{proof}

\subsection{Proof of Theorem \ref{thm:monotone2}}
\label{appdx:proof_of_mono2}
In this section, we analyze Theorem \ref{thm:monotone2}, which establishes the sample complexity and approximation ratio guarantees for the solution obtained by \alglongmono (\algmono). \algmono is an algorithm for the MSMC problem where only noisy queries to $\Delta f$ are available. The corresponding algorithm description is presented in Algorithm \ref{alg:ATG2}.

First of all, we give a brief description of the \algmono algorithm. \algmono shares a similar idea with the \alg algorithm presented in Section \ref{sec:monotone}. Both of the two algorithms utilize \samp to determine if the expectation of the evaluated marginal gain is approximately above a threshold $w$. However, they differ in their error approximation guarantees on the expectation of evaluated marginal gain. Specifically, \alg invokes the \samplong procedure (\samp) with the following inputs: threshold $w$, approximation error bound $\epsilon$, error probability $\frac{2\delta}{3nh'(\alpha)}$ where $h'(\alpha)=\frac{3\log{(3\kappa/\alpha)}}{\alpha}$, random distribution $\mathcal{D}(S,u)$, and upper bound of the noisy marginal gain $R$ as input. Different from the subroutine algorithm \samp in \alg, the worst-case query complexity $N_1$ and confidence interval $C_t$ in \samp are defined as in Theorem \ref{thm:sampling2} with the multiplicative input parameter set to $\alpha/3$.
Therefore, the output of \samp in \algmono satisfies that with high probability, if the output is true, then $(1+\alpha/3)\Delta f(S,u)\geq w-\epsilon$. If the output is false, then $(1-\alpha/3)\Delta f(S,u)\leq w+\epsilon$.


Next, we present the analysis of Theorem \ref{thm:monotone2}.
\begin{theorem}
    \label{thm:monotone2}
    Suppose the noisy marginal gain of any subset $S\subseteq U$ and element $s\in U$ is bounded in $[0,R]$, \algmono makes at most $3n\log(\kappa/\alpha)/\alpha$ calls of \samp. In addition, with probability at least $1-\delta$, the following statements hold:
    \begin{itemize}[noitemsep]
        \item The exact function value of the output solution set $S$ satisfies that $f(S)\geq(1-e^{-1}-\alpha)f(OPT)-2\kappa\epsilon$;
    \item Each call of \samp on input ($w$, $\epsilon$, $\frac{2\delta}{3nh'(\alpha)}$, $\mathcal{D}(S,u)$, $R$) takes at most the minimum between
    \begin{align*}
\frac{9R}{\epsilon\alpha}\log \left(\frac{6nh'(\alpha)}{\delta}\right)
    \end{align*}
    and 
    \begin{align*}
       \frac{36R}{\alpha\phi'(S,u)}\log\left(\frac{36R}{\alpha\phi'(S,u)}\sqrt{\frac{12nh'(\alpha)}{\delta}}\right)
    \end{align*}
    noisy samples. Here $OPT$ is an optimal solution to the MSMC problem, $$\phi'(S,u) = \frac{\epsilon -\alpha\Delta f(S,u)/3+ |w-\Delta f(S,u)|}{2},$$ and $$h'(\alpha)=\frac{3}{\alpha}\log{(\frac{3\kappa}{\alpha})}.$$
    \end{itemize}
    
\end{theorem}


Notice that the sample complexity in Theorem \ref{thm:monotone2} has a dependence of $O(R)$ concerning the order of the parameter $R$, while the sample complexity result in Theorem \ref{mainthm} is $O(R^2)$ in the order of $R$. Consequently, in some applications such as influence maximization, where $R$ can be as large as the size of ground set $n$, Theorem \ref{thm:monotone2} has an advantage in sample complexity compared with Theorem \ref{mainthm}.

Now we present the proof of Theorem \ref{thm:monotone2}. The organization of the proof for Theorem \ref{thm:monotone2} is as follows: we begin by presenting the proof of the Theorem \ref{thm:monotone2}. Then the proofs of two lemmas, Lemma \ref{lem:clean_event_call_to_CS_mono2} and Lemma \ref{lem:mono_iterative_result}, that are used in the proof of Theorem \ref{thm:monotone2} are presented.
\label{appdx:mono2}

\begin{algorithm}[t]
\caption{\alglongmono (\algmono)}\label{alg:ATG2}
 \begin{algorithmic}[1]
 \STATE \textbf{Input:} $\epsilon$, $\delta, \alpha$
 \STATE $N_3\gets \frac{9R}{\epsilon\alpha}\log\frac{6n}{\delta}$
 \FORALL{$s\in U$}
 \STATE $\hat{f}_{N_3}(s) \gets $ sample mean over $N_3$ samples from $\mathcal{D}(\emptyset,s)$ \label{alg:ATG:line:sample-mean2}
 \ENDFOR
  
  \STATE $d:=\max_{s\in U}\hat{f}_{N_3}(s)$, 
 \STATE $w\gets d$, $S\gets \emptyset$
 \WHILE{$w> \frac{\alpha d}{3\kappa}$}
 \FORALL{$u\in U$} \label{line:algmono_loop_start}
\IF{$|S|<\kappa$}
 \STATE thre = \samplong($w$, $\epsilon$, $\frac{2\delta}{3nh'(\alpha)}$, $\mathcal{D}(S,u)$, $R$)
 \IF{thre}
 \STATE $S\gets S\cup \{u\}$
 \ENDIF
  \ENDIF
 \ENDFOR

 \STATE $w=w(1-\alpha/3)$
 \label{line:algmono_loop_end}
 \ENDWHILE
 \STATE \textbf{return} $S$
 \end{algorithmic}
\end{algorithm}

    \begin{proof}
        First, since the number of iterations in the while loop from Line \ref{line:algmono_loop_start} to Line \ref{line:algmono_loop_end} in \algmono (see Algorithm \ref{alg:ATG2}) is upper bounded by $\frac{3}{\alpha}\log\frac{3\kappa}{\alpha}$, \algmono makes at most $\frac{3n}{\alpha}\log\frac{3\kappa}{\alpha}$ calls of \sampnew. Next, we prove the second result in Theorem \ref{thm:monotone2}, which guarantees the upper bound on the required number of samples. By applying Lemma \ref{lem:clean_event_call_to_CS_mono2} on the sampling of the noisy marginal gain of $\Delta f(S,u)$, we can see that with probability at least $1-\delta$, for each call of \sampnew, we have that the number of noisy queries is bounded by the minimum between 
        $\frac{9R}{\epsilon\alpha}\log \left(\frac{6nh'(\alpha)}{\delta}\right)$ and $\frac{36R}{\alpha\phi'(S,u)}\log\left(\frac{36R}{\alpha\phi'(S,u)}\sqrt{\frac{12nh'(\alpha)}{\delta}}\right)$.
        
        Now we prove the first result. Since the proof of the first result is similar to the proof of Theorem \ref{mainthm}, here we provide a proof sketch and omit the details. First of all, by Lemma \ref{lem:mono_iterative_result}, we have
        \begin{align*}
             f(S_{i+1})\geq \frac{1-\alpha}{\kappa}f(OPT)+(1-\frac{1-\alpha}{\kappa})f(S_i)-2\epsilon.
        \end{align*}
        Let us denote the solution set $S$ after the $i$-th element is added as $S_i$. Notice that the result in Lemma \ref{lem:mar_gain} is the same as Lemma \ref{lem:mono_iterative_result}. Therefore, following the same proof as that in Theorem \ref{mainthm}, we would get that if $|S|=\kappa$, then by induction
        \begin{align*}
    f(S_{\kappa})
    &\geq(1-e^{-1}-\alpha)f(OPT)-2\kappa\epsilon.
\end{align*}
If the size of the output solution set $S$ is smaller than $\kappa$, then any element $o\in OPT$ that is not added to $S$ at the last iteration satisfies that
$(1-\alpha/3)\Delta f(S,o)\leq w+\epsilon$. Since at the last iteration $w\leq\frac{\alpha d}{3\kappa}$, and that conditioned on the events in Lemma \ref{lem:clean_event_call_to_CS_mono2}, $d\leq(1+\alpha/3)\max_{s\in U}f(s)+\epsilon$, it follows that 
\begin{align*}
    (1-\alpha/3)\Delta f(S,o)\leq \frac{\alpha }{3\kappa}\{(1+\alpha/3)\max_{s\in U}f(s)+\epsilon\}+\epsilon
\end{align*}
By submodularity and monotonicity of $f$, we have
\begin{align*}
    f(OPT)-f(S)&\leq\sum_{o\in OPT}\Delta f(S,o)\\
    &\leq \frac{\alpha }{3(1-\alpha/3)}\{(1+\alpha/3)\max_{s\in U}f(s)+\epsilon\}\\
    &\qquad+\frac{\kappa\epsilon}{(1-\alpha/3)}\\
    &\leq\alpha\max_{s\in U}f(s)+2\kappa\epsilon\\
    &\leq \alpha f(OPT)+2\kappa\epsilon.
\end{align*}
Then we have $f(S)\geq (1-\alpha)f(OPT)-2\kappa\epsilon$.
    \end{proof}
    The proof of the above Theorem \ref{thm:monotone2} depends on Lemma \ref{lem:mono_iterative_result}. Before proving Lemma \ref{lem:mono_iterative_result}, we first prove the Lemma \ref{lem:clean_event_call_to_CS_mono2}.
 \begin{lemma}
 \label{lem:clean_event_call_to_CS_mono2}
     With probability at least $1-\delta$, the following two events hold.
     \begin{enumerate}
     \item $(1-\alpha/3)\max_{s\in U}f(s)-\epsilon\leq d\leq(1+\alpha/3)\max_{s\in U}f(s)+\epsilon$.
         \item During each call of \sampnew on input ($w$, $\epsilon$, $\frac{2\delta}{3nh'(\alpha)}$, $\mathcal{D}(S,u)$, $R$), if the output is true, then $(1+\alpha/3)\Delta f(S,u)\geq w-\epsilon$. If the output is false, then $(1-\alpha/3)\Delta f(S,u)\leq w+\epsilon$. In addition, the number of samples taken by \sampnew is at most the minimum between
         \begin{align}
    \label{eq:sam_complxt1}
       \frac{9R}{\epsilon\alpha}\log \left(\frac{6nh'(\alpha)}{\delta}\right)
    \end{align}
    and 
    \begin{align}
    \label{eqn:sam_complx2}
        \frac{36R}{\alpha\phi'(S,u)}\log\left(\frac{36R}{\alpha\phi'(S,u)}\sqrt{\frac{12nh'(\alpha)}{\delta}}\right),
    \end{align}
    
    where $\phi'(S,u) = \frac{\epsilon -\alpha\Delta f(S,u)/3+ |w-\Delta f(S,u)|}{2}$, and $h'(\alpha)=\frac{3}{\alpha}\log{(\frac{3\kappa}{\alpha})}$.
     \end{enumerate}
 \end{lemma}
 
 \begin{proof}
      First of all, by applying the inequality in Lemma \ref{lem:chernoff}, we have that for
      fixed element $s\in U$
      \begin{align*}
          P\big(|\hat{f}_{N_3}(s)-f(s)|\geq\frac{\alpha}{3}f(s)+\epsilon\big)\leq\frac{\delta}{3n}.
      \end{align*}
      Taking a union bound over all elements in $U$, it follows that 
      \begin{align*}
          P\big(|\hat{f}_{N_3}(s)-f(s)|\geq\frac{\alpha}{3}f(s)+\epsilon,\forall s\in U\big)\leq\frac{\delta}{3}, 
      \end{align*}
      where $N_3=\frac{9R}{\epsilon\alpha}\log\frac{6n}{\delta}$. 
      Therefore, with probability at least $1-\delta/3$, we have $|\hat{f}_{N_3}(s)-f(s)|\leq\frac{\alpha}{3}f(s)+\epsilon$ for each $s\in U$. Denote $s_1=\arg\max_{s\in U}\hat{f}_{N_3}(s)$ and $s_2=\arg\max_{s\in U}{f}(s)$. It follows that with probability at least $1-\delta/3$, we have that 
      \begin{align*}
          d=\hat{f}_{N_3}(s_1)\leq (1+\alpha/3)f(s_1)+\epsilon\leq(1+\alpha/3)f(s_2)+\epsilon,
      \end{align*}
      and that 
      \begin{align*}
          d=\hat{f}_{N_3}(s_1)\geq \hat{f}_{N_3}(s_2)\geq(1-\alpha/3)f(s_2)-\epsilon.
      \end{align*}
      Since $d=\max_{s\in U}\hat{f}_{N_3}(s)=\hat{f}_{N_3}(s_1)$ and $f(s_2)=\max_{s\in U}f(s)$, the first result holds with probability at least $1-\delta/3$.

      Next, we prove the second result. For
      each call of the sampling algorithm \sampnew with fixed input ($w$, $\epsilon$, $\frac{2\delta}{3nh'(\alpha)}$, $\mathcal{D}(S,u)$, $R$), and given that $N_1$ and $C_t$ are defined in accordance with Theorem \ref{thm:sampling2} with the multiplicative error parameter set to $\alpha/3$, we can leverage the second result in Theorem \ref{thm:sampling2}. Consequently, with  probability at least $1-\frac{2\delta}{3nh'(\alpha)}$, the following two things hold: 
      \begin{enumerate}
          \item If the output of \sampnew is true, then $(1+\alpha/3)\Delta f(S,s)\geq w-\epsilon$. If the output is false, then $(1-\alpha/3)\Delta f(S,s)\leq w+\epsilon$.
          \item The number of noisy queries is bounded by the minimum between (\ref{eq:sam_complxt1}) and (\ref{eqn:sam_complx2}) in the lemma.
      \end{enumerate}
       Since there are at most $\frac{\log(3\kappa/\alpha)}{\log\frac{1}{1-\alpha/3}}\leq h'(\alpha)=\frac{3}{\alpha}\log\frac{3\kappa}{\alpha}$ number of iterations in \algmono, there are at most $nh'(\alpha)$ calls of \sampnew. Therefore, by taking the union bound we have that with probability at least $1-2\delta/3$, the two events defined above hold for all calls to \sampnew during \algmono. By taking the union bound again, we have that with probability at least $1-\delta$, the two results in the lemma both hold.
    \end{proof}
    Now we prove the Lemma \ref{lem:mono_iterative_result}.
    
    \begin{lemma}
        \label{lem:mono_iterative_result}
         Assume the events defined in Lemma \ref{lem:clean_event_call_to_CS_mono2} hold during \algmono. Then for any element $s$ that is added to the solution set $S$, the following statement holds.
    \begin{align*}
        \Delta f(S,s)\geq \frac{1-\alpha}{\kappa}(f(OPT)-f(S))-2\epsilon.
    \end{align*}
    \end{lemma}
    \begin{proof}   
    At the first iteration, if an element $s$ is added to the solution set, it holds by Lemma \ref{lem:clean_event_call_to_CS_mono2} that $(1+\frac{\alpha}{3})\Delta f(S,s)\geq w-\epsilon$.
   Since at the first iteration $w=d$ and $d\geq (1-\alpha/3)\max_{s\in U}f(s)-\epsilon$. It follows that $\Delta f(S,s)\geq \frac{1-\alpha/3}{1+\alpha/3}\max_{s\in U}f(s)-\frac{2\epsilon}{1+\alpha/3}\geq(1-\alpha)\max_{s\in U}f(s)-2\epsilon$.
    By submodularity we have that $\kappa\max_{s\in U}f(s)\geq f(OPT)$. Therefore, $\Delta f(S,s)\geq \frac{1-\alpha}{\kappa}(f(OPT)-f(S))-2\epsilon$.
    
    At iteration $i$ where $i>1$, if an element $o\in OPT$ is not added to the solution set, then it is not added to the solution set at the last iteration, where the threshold is $\frac{w}{1-\alpha/3}$. By Lemma \ref{lem:clean_event_call_to_CS_mono2}, we have
    $(1-\alpha/3)\Delta f(S,o)\leq \frac{w}{1-\alpha/3}+\epsilon$.
    For any element $s$ that is added to the solution at iteration $i$, by Lemma \ref{lem:clean_event_call_to_CS_mono2} it holds that $(1+\alpha/3)\Delta f(S,s)\geq w-\epsilon$. Therefore, we have
    \begin{align*}
        \Delta f(S,s)&\geq \frac{w-\epsilon}{1+\alpha/3} \\
        &\geq\frac{(1-\alpha/3)^2\Delta f(S,o)-(1-\alpha/3)\epsilon-\epsilon}{1+\alpha/3} \\
        &\geq (1-\alpha)\Delta f(S,o)-2\epsilon. 
    \end{align*}
    By submodularity, it holds that $\Delta f(S,s)\geq (1-\alpha)\frac{f(OPT)-f(S)}{\kappa}-2\epsilon$.
\end{proof}

\section{Appendix for Section \ref{sec:nonmono}}
In this section, we present the proof of Theorem \ref{thm:nonmono}.

\noindent\textbf{Theorem \ref{thm:nonmono}. }\textit{\cdg makes $n$ calls of \samp. In addition, with probability at least $1-\delta$, the following statements hold:
    \begin{enumerate}
        \item The exact function value of the output solution set $S$ satisfies that $f(S)\geq\frac{f(OPT)}{3}-\epsilon$;
        \item Each call of \samp on input $(0,\frac{3\epsilon}{n},\frac{\delta}{n},\mathcal{D}_{X_i},\sqrt{2}R)$ takes at most the minimum between
    \begin{align*}
        \left\{\frac{4n^2R^2}{9\epsilon^2}\log \left(\frac{4n}{\delta}\right),\frac{16R^2}{\phi^2_i}\log\left(\frac{32R^2}{\phi^2_i}\sqrt{\frac{2n}{\delta}}\right)\right\}
    \end{align*}
    noisy samples. Here $OPT$ is an optimal solution to the USM problem, and
    \begin{align*}
        \phi_{i} &:= \frac{3\epsilon/n + |\mE X_i|}{2}\\
        &=\frac{3\epsilon/n + |\Delta f(A_{i-1},u_i)+\Delta f(B_{i-1}/\{u_i\},u_i)|}{2}.
    \end{align*}
    \end{enumerate}
   }
Notice that conditioned on the solution set $A_{i-1}$ and $B_{i-1}$, the random variables $\widetilde{\Delta f}(A_{i-1}, u_i)$ and $\widetilde{\Delta f}(B_{i-1}/\{u_i\},u_i)$ are $R$-sub-Gaussian. Therefore, $X_i:=\widetilde{\Delta f}(A_{i-1}, u_i)+\widetilde{\Delta f}(B_{i-1}/\{u_i\},u_i)$ is $\sqrt{2}R$-sub-Gaussian, the second result is implied by applying Theorem \ref{thm:sampling} immediately. To prove the first result in Theorem \ref{thm:nonmono}, we need the following lemma.
\begin{lemma}
\label{lem:nonmono}
With probability at least $1-\frac{\delta}{n}$, the $i$-th call of \samp satisfies the following inequality 
\begin{align}
\label{eqn:nonmono_iter}
     f(A_{i-1}&\cup OPT_{i-1})-f(A_{i}\cup OPT_{i})\leq \nonumber\\
    &[f(A_i)-f(A_{i-1})] + [f(B_i)-f(B_{i-1})]+\frac{3\epsilon}{n}.
\end{align}
where $OPT_i$ is the set of all elements from $OPT$ that arrives after the $i$-th iteration.
\end{lemma}
\begin{proof}
    From the statement of the algorithm, we know that the element $u_i$ is added to the solution if and only if the output of \samp is true. By applying the results in Theorem \ref{thm:sampling}, we have that for each fixed $i$, with probability at least $1-\delta/n$ if $u_i$ is added, then $\Delta f(A_{i-1},u_i)\geq -\Delta f(B_{i-1}/\{u_i\},u_i)-\frac{3\epsilon}{n}$. Otherwise, $\Delta f(A_{i-1},u_i)\leq -\Delta f(B_{i-1}/\{u_i\},u_i)+\frac{3\epsilon}{n}$. Let us denote the above event as $\mathcal{E}_i$, we discuss the following four cases in our analysis
    \begin{enumerate}
        \item If $u_i\in A_i$, and $u_i\in OPT$, then 
        \begin{align*}
             f(A_{i-1}\cup OPT_{i-1})-&f(A_{i}\cup OPT_{i})=0
        \end{align*}
        Notice that $u_i\in A_i$, then conditioned on $\mathcal{E}_i$, we have $\Delta f(A_{i-1},u_i)\geq -\Delta f(B_{i-1}/\{u_i\},u_i)-\frac{3\epsilon}{n}$. By submodularity, $\Delta f(B_{i-1}/\{u_i\},u_i)\leq\Delta f(A_{i-1},u_i)$. Then it follows that $\Delta f(A_{i-1}, u_i)+\frac{3\epsilon}{2n}\geq 0$. Therefore, the term on the right-hand side of $(\ref{eqn:nonmono_iter})$ satisfies 
        \begin{align*}
            [f(A_i)-f(A_{i-1})]& + [f(B_i)-f(B_{i-1})]+\frac{3\epsilon}{n}\\
            &=\Delta f(A_{i-1}, u_i)+\frac{3\epsilon}{n}\geq 0.
        \end{align*}

        \item If $u_i\in A_i$, and $u_i\notin OPT$, then 
        \begin{align*}
            f(A_{i-1}\cup OPT_{i-1})-&f(A_{i}\cup OPT_{i})\\
            &=-\Delta f(A_{i-1}\cup OPT_i,u_i)\\
            &\leq -\Delta f(B_{i-1}/\{u_i\},u_i),
        \end{align*}
  where the inequality is obtained by submodularity. The right-hand side in (\ref{eqn:nonmono_iter}) is 
        \begin{align*}
            [f(A_i)-f(A_{i-1})]& + [f(B_i)-f(B_{i-1})]+\frac{3\epsilon}{n} \\
            &= \Delta f(A_{i-1}, u_i)+\frac{3\epsilon}{n}.
        \end{align*}
        Notice that $u_i\in A_i$, then conditioned on $\mathcal{E}_i$, we have $\Delta f(A_{i-1},u_i)\geq -\Delta f(B_i/\{u_i\},u_i)-\frac{3\epsilon}{n}$. Therefore, 
        \begin{align*}
            [f(A_i)-f(A_{i-1})]& + [f(B_i)-f(B_{i-1})]+\frac{3\epsilon}{n}\\
            &=\Delta f(A_{i-1}, u_i)+\frac{3\epsilon}{n}\\
            &\geq -\Delta f(B_{i-1}/\{u_i\},u_i).
        \end{align*}

        \item If $u_i\notin A_i$, and $u_i\notin OPT$, then 
        \begin{align*}
            f(A_{i-1}\cup OPT_{i-1})-&f(A_{i}\cup OPT_{i})=0.
        \end{align*}
        Similarly as the first case, we have that $-\Delta f(B_{i-1}/\{u_i\},u_i)\geq\frac{3\epsilon}{2n}$. Since the right-hand side is $-\Delta f(B_{i-1}/\{u_i\},u_i)+\frac{3\epsilon}{n}$, the inequality holds.
        \item If $u_i\notin A_i$, and $u_i\in OPT$, then
        \begin{align*}
            &f(A_{i-1}\cup OPT_{i-1})-f(A_{i}\cup OPT_{i})\\
            &=\Delta f(A_{i-1}\cup OPT_i,u_i)\leq \Delta f(A_{i-1},u_i),
        \end{align*}
        where the inequality holds by submodularity. Conditioned on the event $\mathcal{E}_i$, it follows that $\Delta f(A_{i-1},u_i)\leq -\Delta f(B_i/\{u_i\},u_i)+\frac{3\epsilon}{n}$. Since the right-hand side is
        \begin{align*}
            [f(A_i)-f(A_{i-1})]& + [f(B_i)-f(B_{i-1})]+\frac{3\epsilon}{n}\\
            &=-\Delta f(B_i/\{u_i\},u_i)+\frac{3\epsilon}{n},
        \end{align*}
the result is proved.
    \end{enumerate}
\end{proof}

Now we prove Theorem \ref{thm:nonmono}. 
\begin{proof}
    Define the event 
    \begin{align*}
        \mathcal{F}_i&=\{f(A_{i-1}\cup OPT_{i-1})-f(A_{i}\cup OPT_{i})\leq \\
    &[f(A_i)-f(A_{i-1})] + [f(B_i)-f(B_{i-1})]+\frac{3\epsilon}{n}\}.
    \end{align*}
    From Lemma \ref{lem:nonmono} and by taking the union bound, it follows that 
    \begin{align*}
        P(\mathcal{F}_i,\forall i\in[n])\geq 1-\delta
    \end{align*}
    Therefore, with probability at least $1-\delta$, $\mathcal{F}_i$ holds for all $i$. Then by summing over all $i$, we would get
    \begin{align*}
     \sum_{i=1}^nf(A_{i-1}&\cup OPT_{i-1})-f(A_{i}\cup OPT_{i})\leq \\
    &\sum_{i=1}^n\{[f(A_i)-f(A_{i-1})] \\
    &+ [f(B_i)-f(B_{i-1})]\}+3\epsilon.
\end{align*}
    It follows that 
    \begin{align*}
     f( OPT_{0})&-f(A_{n})\leq \\
   & [f(A_n)-f(A_{0})] + [f(B_n)-f(B_{0})]\}+3\epsilon.
\end{align*}
Since the submodular function is nonnegative, and that $f(A_n)=f(B_n)$, $OPT_0=OPT$, it follows that $f(A)\geq f(OPT)/3-\epsilon$.
\end{proof}

\section{Appendix for Section \ref{sec:matroid}}
\label{appdx:continuous}
In this section, we present supplementary material to Section \ref{sec:matroid}. In particular, we present the comparison of the result in Theorem \ref{thm:continuous} to the Accelerated Continuous Greedy algorithm (ACG) in \cite{badanidiyuru2014fast}. Then in Section \ref{appdx:continuous}, we provide detailed proof of Theorem \ref{thm:continuous}.

\subsection{Comparison of \contialg with Accelerated Continuous Greedy algorithm}
\label{appdx:comparison_to_ACG}
In this section, we compare the results of Theorem \ref{thm:continuous} and the Accelerated Continuous Greedy algorithm (ACG) as presented in \cite{badanidiyuru2014fast}.
\begin{enumerate}
    \item First of all, we consider the case where we have exact access to the value oracle. In this case, we can get that $\widetilde{\Delta f}(S,s)=\Delta f(S,s)\leq \max_{s\in S}f(s)$ for any subset $S\subseteq U$ and element $s\in U$. This implies that $R$ can be set to be $\max_{s\in S}f(s)$. Consequently, from Theorem \ref{thm:continuous}, the output solution set of \contialg satisfies that $f(S)\geq (1-1/e-O(\epsilon))f(OPT)$, which aligns with the approximation ratio presented in \cite{badanidiyuru2014fast}. For the result on sample complexity, notice that each call of \samp takes at most $\min\{O(\frac{\kappa}{\epsilon^2}\log\frac{n}{\delta\epsilon}), O(\frac{\kappa}{\epsilon\phi_X''}\log\frac{n}{\delta\epsilon\phi_X''})\}$ number of samples, where the first result is obtained by considering the worst case sample complexity of a fixed $\epsilon$-approximation. Since there are at most $\frac{3n}{\epsilon^2}\log\frac{\kappa}{\epsilon}$ calls of \samp during \contialg, if we only consider the worst-case sample complexity, the total required sample complexity is at most $O(\frac{\kappa n}{\epsilon^3}\log^2\frac{n}{\epsilon})$ for \contialg. This matches the result in \cite{badanidiyuru2014fast}. In this sense, we improve the sample complexity when reduced to the case of assuming an exact oracle to the marginal gains.
    
    \item On the other hand, from Theorem \ref{thm:continuous}, we can see that even if the access to $\Delta f$ is noisy, as long as the upper bound on the noisy marginal gain $R$ is less than $ f(OPT)$, the above analysis on sample complexity and approximation ratio holds. Hence, we can conclude that compared to access to an exact value oracle, the assumption of access to noisy marginal gain does not lead to additional sample complexity or a deterioration in the approximation ratio when compared to the scenario with an exact value oracle.
\end{enumerate}
\subsection{Proof of Theorem \ref{thm:continuous}}
\label{appdx:proof_of_conti}
In this section, we present the detailed proof of Theorem \ref{thm:continuous} about our algorithm \contialg.

\noindent\textbf{Theorem \ref{thm:continuous}. }\textit{\contialg makes at most $\frac{3n}{\epsilon^2}\log\frac{3\kappa}{\epsilon}$ calls of \sampnew. In addition, with probability at least $1-\delta$, the following statements hold:
\begin{itemize}
        \item The output fractional solution $\vect{x}$ achieves the approximation guarantee of $\vect{F}(\vect{x})\geq(1-e^{-1}-2\epsilon)f(OPT)-R\epsilon$.
        \item Each call of \sampnew on input  ($w$, $\frac{\epsilon R}{2\kappa}$, $\frac{\delta\epsilon}{2nh'(\epsilon)}$, $\mathcal{D}_X$, $R$) requires at most the minimum between
    \begin{align*}
      \frac{18\kappa}{\epsilon^2}\log \left(\frac{8nh'(\epsilon)}{\delta\epsilon}\right)
    \end{align*}
    and
    \begin{align*}
        \frac{36R}{\epsilon\phi''_X}\log\left(\frac{144R}{\epsilon\phi''_X}\sqrt{\frac{nh'(\epsilon)}{\delta\epsilon}}\right)
    \end{align*}
    noisy queries to the marginal gain. Here $OPT$ is an optimal solution to the MSMM problem, $\phi''_X = \frac{\frac{\epsilon R}{2\kappa} -\epsilon\mathbb{E}X  /{3}+ |w-\mathbb{E}X|}{2}$, and $h'(\epsilon)=\frac{3}{\epsilon}\log{(\frac{3\kappa}{\epsilon})}$.
    \end{itemize}
}
    \begin{proof}
    The second result on the sample complexity of calling the subroutine algorithm \sampnew can be obtained immediately by applying the second result in (\ref{lem:item_samp_continuous}) in Lemma \ref{lem:clean_event_continuous}. Here we prove the first result in the theorem. Let us denote the fractional solution at time step $t$ as $\vect{x}_t$. From Lemma \ref{lem:Continuous_decreading_threshold}, it follows that conditioned on the events in Lemma \ref{lem:clean_event_continuous}, we have
    \begin{align*}
        \vect{F}(\vect{x}_{t+1})-\vect{F}(\vect{x}_{t})&\geq\epsilon(1-\epsilon)f(OPT)\\
        &\qquad-\epsilon(1-\epsilon)\vect{F}(\vect{x}_{t+1})-\epsilon^2 R.
    \end{align*} 
    It then follows that
    \begin{align*}
        \vect{F}(\vect{x}_{t+1})&\geq\frac{\vect{F}(\vect{x}_{t})+\epsilon(1-\epsilon)f(OPT)-\epsilon^2 R}{1+\epsilon(1-\epsilon)}\\
        &\geq(1-\epsilon)\vect{F}(\vect{x}_{t})+\epsilon(1-\epsilon)^2f(OPT)-\epsilon^2 R
    \end{align*}
   Since there are $1/\epsilon$ iterations in \contialg, the output $\vect{x}$ satisfies that $\vect{x}=\vect{x}_{1/\epsilon}$. By applying induction to the above inequality, we would get
    \begin{align*}
        \vect{F}(\vect{x}_{1/\epsilon})
        &\geq(1-(1-\epsilon)^{1/\epsilon})\{(1-\epsilon)^2f(OPT)-\epsilon R\}\\
        &\geq(1-1/e)\{(1-\epsilon)^2f(OPT)-\epsilon R\}\\
        &\geq(1-1/e-2\epsilon)f(OPT)-\epsilon R.
    \end{align*}
    \end{proof}

\begin{lemma}
\label{lem:clean_event_continuous}
    With probability at least $1-\delta$, the following two events hold.
    \begin{enumerate}
     \item $(1-\epsilon/3)\max_{s\in U}f(s)-\frac{R\epsilon}{2\kappa}\leq d\leq(1+\epsilon/3)\max_{s\in U}f(s)+\frac{R\epsilon}{2\kappa}$.
         \item During each call of \sampnew on the input ($w$, $\frac{\epsilon R}{2\kappa}$, $\frac{\delta\epsilon}{2nh'(\epsilon)}$, $\mathcal{D}_X$, $R$, $\epsilon/3$
        ) with the evaluated random variable being $X=\widetilde{\Delta f}(S(\vect{x}+\epsilon\vect{1}_B),u)$ where $\vect{x}$ is the fractional solution , $B$ is the set of coordinates and $u$ is an element in $U$, the results in Theorem \ref{thm:sampling2} holds. I.e.,
        \begin{enumerate}
            \item\label{lem:item_samp_continuous} \sampnew takes at most the minimum between 
            \begin{align*}
      \frac{18\kappa}{\epsilon^2}\log \left(\frac{8nh'(\epsilon)}{\delta\epsilon}\right)
    \end{align*}
    and
    \begin{align*}
        \frac{36R}{\epsilon\phi''_X}\log\left(\frac{144R}{\epsilon\phi''_X}\sqrt{\frac{nh'(\epsilon)}{\delta\epsilon}}\right).
    \end{align*}
\item\label{lem:item_approx_continuous} If the output is true, then $$(1+\epsilon/3)\mE \widetilde{\Delta f}(S(\vect{x}+\epsilon\vect{1}_B),u)\geq w-\frac{\epsilon R}{2\kappa}.$$ 
         If the output is false, then $$(1-\epsilon/3)\mE \widetilde{\Delta f}(S(\vect{x}+\epsilon\vect{1}_B),u)\leq w+\frac{\epsilon R}{2\kappa}.$$
        \end{enumerate}

     \end{enumerate}
\end{lemma}
 \begin{proof}
     First of all, by applying the inequality in Lemma \ref{lem:chernoff}, we have that for each fixed $s\in U$, after taking $N_4=\frac{18\kappa}{\epsilon^2}\log\frac{4n}{\delta}$ number of samples, it follows that
     \begin{align*}
         P\big(|\hat{f}_{N_4}(s)-f(s)|\geq\frac{\epsilon}{3}f(s)+\frac{R\epsilon}{2\kappa}\big)\leq\frac{\delta}{2n}.
     \end{align*}
     Taking a union bound over all elements in $U$, it follows that 
      \begin{align*}
          P\big(|\hat{f}_{N_4}(s)-f(s)|\geq\frac{\epsilon}{3}f(s)+\frac{R\epsilon}{2\kappa},\forall s\in U\big)\leq\frac{\delta}{2}.
      \end{align*}
      Following the similar idea as in the proof of the Lemma \ref{lem:clean_event_samp2}, we can prove the first result. 
      
      Now we start to prove the second result. For each fixed call of \sampnew  with input ($w$, $\frac{\epsilon R}{2\kappa}$, $\frac{\delta\epsilon}{2nh'(\epsilon)}$, $\mathcal{D}_X$, $R$, $\epsilon/3$
        ), by applying the results in Theorem \ref{thm:sampling2}, we have that with probability at least $1-\frac{\delta\epsilon}{2nh'(\epsilon)}$, both the statements about the sample complexity in (\ref{lem:item_samp_continuous}) and approximation guarantee in (\ref{lem:item_approx_continuous}) in the lemma holds. Since there are $1/\epsilon$ calls of the \contisublong and each \contisublong makes at most $nh'(\epsilon)$ calls of the \sampnew algorithm, there are at most $nh'(\epsilon)/\epsilon$ calls of the \sampnew algorithm. By taking the union bound, we can prove that with probability at least $1-\delta/2$, the second results hold. By taking the union bound again, we can see that with probability at least $1-\delta$, all of the results in the lemma hold.
 \end{proof}  
 \begin{lemma}
     \label{lem:Continuous_decreading_threshold}
     Conditioned on the two events defined in Lemma \ref{lem:clean_event_continuous}, we have that during each implementation of \contisublong, the output coordinate set $B$ satisfies that
     \begin{align*}
         \vect{F}(\vect{x}+\epsilon\vect{1}_B)-\vect{F}(\vect{x})&\geq\epsilon(1-\epsilon)\{f(OPT)-\vect{F}(\vect{x}+\epsilon\vect{1}_B)\}\\
         &\qquad-\epsilon^2R.
     \end{align*}
 \end{lemma}
 \begin{proof}
     Here we denote the output solution set as $B=\{b_1,b_2,...,b_{\kappa}\}$ where $b_i$ is the $i$-th element that is added to set $B$. Here if $|B|<\kappa$, then for any $i>|B|$, $b_i$ is defined as a dummy variable. Since $\mathcal{M}$ is a matroid, there exists a permutation of the optimal solution $OPT=\{o_1,o_2,...,o_\kappa\}$ such that $B_{i-1}\cup \{o_i\}\in\mathcal{M}$ for each $i\in[\kappa]$. For notation simplicity, we also define $G(\vect{x},u)=\mE\widetilde{\Delta f}(S(\vect{x}),u) $. First of all, we prove the following claim: for each $i\in[\kappa]$, we have that
     \begin{align*}
         G(\vect{x}+\epsilon\vect{1}_{B_{i-1}},b_i)\geq(1-\epsilon)G(\vect{x}+\epsilon\vect{1}_{B_{i-1}},o_i)-\frac{\epsilon R}{\kappa}
     \end{align*}
     The proof is as follows: if the element $b_i$ is added at the first iteration, then from Lemma \ref{lem:clean_event_continuous}, we have that 
$(1+\epsilon/3)G(\vect{x}+\epsilon\vect{1}_{B_{i-1}},b_i)\geq w-\frac{\epsilon R}{2\kappa}$. Since the threshold at the first iteration is $w=d$, and $d\geq (1-\epsilon/3)\max_{s\in U}f(s)-\frac{R\epsilon}{2\kappa}$ according to the first result in Lemma \ref{lem:clean_event_continuous}, then
\begin{align*}
    (1+\epsilon/3)G(\vect{x}+\epsilon\vect{1}_{B_{i-1}},b_i)\geq (1-\epsilon/3)\max_{s\in U}f(s)-\frac{\epsilon R}{\kappa}.
\end{align*}
Since $\max_{s\in U}f(s)\geq\max_{o\in OPT}f(o)\geq G(\vect{x}+\epsilon\vect{1}_{B_{i-1}},o_i) $, $\forall i\in[\kappa]$, it then follows that 
\begin{align*}
    G(\vect{x}+\epsilon\vect{1}_{B_{i-1}},b_i)\geq (1-\epsilon)G(\vect{x}+\epsilon\vect{1}_{B_{i-1}},o_i)-\frac{\epsilon R}{\kappa}.
\end{align*}
If $b_i$ is not a dummy variable and is not added in the first iteration, we can see that $(1+\epsilon/3)G(\vect{x}+\epsilon\vect{1}_{B_{i-1}},b_i)\geq w-\frac{R\epsilon}{2\kappa}$. Since the element $o_i$ is not added to $B$, it is not added at the last iteration. By the construction of $OPT$, we have that $B_{i-1}\cup\{o_i\}\in\mathcal{M}$. Therefore, 
\begin{align*}
(1-\epsilon/3)G(\vect{x}+\epsilon\vect{1}_{B_{i-1}},o_i)\leq \frac{w}{1-\epsilon/3}+\frac{R\epsilon}{2\kappa}.
\end{align*}
Then
\begin{align*}
    G(\vect{x}+\epsilon\vect{1}_{B_{i-1}},b_i)&
    \geq \frac{(1-\epsilon/3)^2G(\vect{x}+\epsilon\vect{1}_{B_{i-1}},o_i)}{1+\epsilon/3}\\
    &\qquad-\frac{(1-\epsilon/3)\epsilon R}{2(1+\epsilon/3)\kappa}-\frac{R\epsilon}{2(1+\epsilon/3)\kappa}\\
    &\geq(1-\epsilon)G(\vect{x}+\epsilon\vect{1}_{B_{i-1}},o_i)-\frac{\epsilon R}{\kappa}.
\end{align*}
Next, we consider the case where $b_i$ is a dummy variable. In this case $ G(\vect{x}+\epsilon\vect{1}_{B_{i-1}},b_i)=0$. Since $o_i$ is not added,
\begin{align*}
    (1-\epsilon/3)G(\vect{x}+\epsilon\vect{1}_{B_{i-1}},o_i)\leq \frac{\epsilon d}{3\kappa}+\frac{R\epsilon}{2\kappa}.
\end{align*}
Since $d\leq(1+\epsilon/3)\max_{s\in U}f(s)+\frac{R\epsilon}{2\kappa}\leq(1+\epsilon/3)R+\frac{R\epsilon}{2\kappa}$. Notice that when $\epsilon>0.5$, the approximation guarantee in Theorem \ref{thm:continuous} is trivial. Therefore, here we can assume $\epsilon\leq0.5$, which implies that $d\leq 3R/2$. Then we have that
\begin{align*}
    (1-\epsilon/3)G(\vect{x}+\epsilon\vect{1}_{B_{i-1}},o_i)\leq \epsilon R/\kappa.
\end{align*}
Therefore,
\begin{align*}
    G(\vect{x}+\epsilon\vect{1}_{B_{i-1}},b_i)&=0\\
    &\geq(1-\epsilon/3)G(\vect{x}+\epsilon\vect{1}_{B_{i-1}},o_i)-\epsilon R/\kappa.
\end{align*}
With this claim, we can prove the results of the lemma.
\begin{align*}
    \vect{F}(\vect{x}+\epsilon\vect{1}_B)-\vect{F}(\vect{x})&=\sum_{i=1}^\kappa\vect{F}(\vect{x}+\epsilon\vect{1}_{B_i})-\vect{F}(\vect{x}+\epsilon\vect{1}_{B_{i-1}})\\
    &=\sum_{i=1}^\kappa \epsilon\cdot\frac{\partial \vect{F}}{\partial b_i}\big|_{x=\vect{x}+\vect{1}_{B_{i-1}}}\\
    &\geq \epsilon\sum_{i=1}^\kappa\mE \Delta f(S(\vect{x}+\epsilon\vect{1}_{B_{i-1}}),b_i)\\
    &= \epsilon\sum_{i=1}^\kappa G(\vect{x}+\epsilon\vect{1}_{B_{i-1}},b_i).
\end{align*}
Here the last equality comes from the fact that $\mathbb{E}\Delta f(S(\vect{x}),u)=\mathbb{E}\widetilde{\Delta f}(S(\vect{x}),u)$.
By the claim, it follows that
\begin{align*}
    \vect{F}(\vect{x}+\epsilon\vect{1}_B)-\vect{F}(\vect{x})
    &\geq\epsilon\sum_{i=1}^\kappa (1-\epsilon)G(\vect{x}+\epsilon\vect{1}_{B_{i-1}},o_i)-\epsilon^2R\\
    &=\epsilon(1-\epsilon)\sum_{i=1}^\kappa \mE \Delta f(S(\vect{x}+\epsilon\vect{1}_{B_{i-1}}),o_i)\\
    &\qquad-\epsilon^2R\\
    &\geq\epsilon(1-\epsilon)\sum_{i=1}^\kappa \mE \Delta f(S(\vect{x}+\epsilon\vect{1}_{B}),o_i)\\
    &\qquad-\epsilon^2R\\
    &\geq\epsilon(1-\epsilon)\{f(OPT)-\vect{F}(\vect{x}+\epsilon\vect{1}_B)\}\\
    &\qquad-\epsilon^2R.
\end{align*}
Here the second and third inequality are due to submodularity and monotonicity.
 \end{proof}
 
     

\section{Technical Lemmas}

\begin{lemma}[Hoeffding's Inequality]
    \label{hoeffding}
    Let $X_1,...,X_N$ be independent random variables such that $X_i$ is $R$-sub-Gaussian and $\mathbb{E}[X_i]=\mu$ for all $i$. Let $\overline{X}=\frac{1}{N}\sum_{i=1}^NX_i$. Then for any $t>0$,
    \begin{align*}
        P(|\overline{X}-\mu|\geq t)\leq 2\exp\{-\frac{Nt^2}{2R^2}\}.
    \end{align*}
\end{lemma}
\begin{lemma}[Relative $+$ Additive Chernoff Bound (Lemma 2.3 in \cite{badanidiyuru2014fast})]
\label{lem:chernoff}
    Let $X_1,...,X_N$ be independent random variables such that for each $i$, $X_i\in[0,R]$ and $\mathbb{E}[X_i]=\mu$ for all $i$. Let $\widehat{X}_N=\frac{1}{N}\sum_{i=1}^NX_i$. Then 
    \begin{align*}
        P(|\widehat{X}_N-\mu|> \alpha\mu+\epsilon)\leq 2\exp\{-\frac{N\alpha\epsilon}{3R}\}.
    \end{align*}
\end{lemma}


\begin{lemma}
    Let $X_1,...,X_N$ be independent random variables such that $X_i\in[0,R]$ and $\mathbb{E}[X_i]=\mu$ for all $i$. Let $\overline{X}=\frac{1}{N}\sum_{i=1}^NX_i$. Then for any $t>0$ and $\delta>0$, if
    \begin{align*}
        N\geq \frac{R^2\ln(1/\delta)}{t^2},
    \end{align*}
    then $P(|\overline{X}-\mu|\geq t)\leq\delta.$
\end{lemma}
\begin{proof}
    This result follows easily from Hoeffding's Inequality.
\end{proof}

\begin{lemma}
    Let $X_1,...,X_N$ be independent random variables such that $X_i\in[0,R]$ and $\mathbb{E}[X_i]=\mu$ for all $i$. Let $\overline{X}=\frac{1}{N}\sum_{i=1}^NX_i$. Then for any $\delta>0$,
    if
    \begin{align}
        c\geq R\sqrt{\frac{\ln(2/\delta)}{2N}},
    \end{align}
    it is the case that
    \begin{align*}
        P(\mu\in[\overline{X}-c,\overline{X}+c])\leq\delta.
    \end{align*}
\end{lemma}
\begin{proof}
    This result follows easily from Hoeffding's Inequality.
\end{proof}

\begin{lemma}
    \label{lem:logx_over_x}
    Suppose $x\in\mathbb{R}$ and $x\geq 2$, if we have $x\geq\frac{2}{a}\log\frac{2}{a}$, then it holds that 
    \begin{align*}
        \frac{\log x}{x}\leq a
    \end{align*}
\end{lemma}
\begin{proof}
    Since $y=\frac{\log x}{x}$ is decreasing when $x\geq 2$, if $x>\frac{2}{a}\log\frac{2}{a}$, then we have 
    \begin{align*}
        \frac{\log x}{x}< \frac{a}{2} \cdot \frac{\log (\frac{2}{a}\log\frac{2}{a})}{\log \frac{2}{a}}\leq a.
    \end{align*}
\end{proof}

\section{Additional Experiments}
In this section, we present some additional details of our experiments.  In particular, we present additional
detail about the experimental setup in Section \ref{appdx:exp_setup}. 
Next, we present the additional experimental results in Section \ref{results}.

\subsection{Additional Experimental Setup}
\label{appdx:exp_setup}
First of all, we provide details about the two applications used to evaluate our algorithms. The two applications considered here are noisy data summarization as presented in Section \ref{sec:data_summarization} and influence maximization in Section \ref{influence}.

\subsubsection{Noisy Data Summarization}
\label{sec:data_summarization}
In data summarization, $U$ is a dataset that we wish to summarize by choosing a subset of $U$ of cardinality at most $\kappa$. The objective function $f:2^U\to\mathbb{R}_{\geq 0}$ takes a subset $X\subseteq U$ to a measure of how well $X$ summarizes the entire dataset $U$, and in many cases is monotone and submodular \cite{tschiatschek2014learning}. However, in real instances of data summarization, we may not have access to an exact measure $f$ of the quality of a summary, but instead, we may have authentic human feedback which is modeled as noisy queries to some underlying monotone and submodular function \cite{singla2016noisy}.

Motivated by this, we run our experiments using instances of noisy data summarization. Our underlying monotone submodular function $f$ is defined as follows: $U$ is assumed to be a labeled dataset, e.g. images tagged with descriptive words, and for any $X\subseteq U$, $f$ takes $X$ to the total number of tags represented by at least one element in $X$ \cite{crawford2023scalable}. Notice that this is essentially the instance of set cover.

\subsubsection{Influence Maximization}
\label{influence}
Another application is the influence maximization problem in large-scale networks \cite{kempe2003maximizing}. In this application, the universe is the set of users in the social network, and the objective is to choose a subset of users to seed with a product to advertise in order to maximize the spread throughout the network. The marginal gain of adding an element $s$ to set $S$ is defined as $\Delta f(S,s):=\mathbb{E}_{\bf{w}\sim \mathcal{D}(\bar{\bf{w}})}\Delta f (S,s;\bf{w})$, where $\bf{w}$ is the noisy realization of the graph from some unknown distribution $\mathcal{D}(\bar{\bf{w}})$, and $\Delta f (S,s;\textbf{w}) = f(S\cup\{s\};\textbf{w})- f(S;\bf{w})$. In a noisy graph realization with parameter $\bf{w}$, $f(S;\bf{w})$ is the number of elements influenced by the set $S$ under some influence cascade model. It is \#P-hard to evaluate the objective in influence maximization \cite{chen2010scalable}. Many of the previous works \cite{chen2009efficient} assume the entire graph can be stored by the algorithm and the influence cascade model is known. The algorithm first samples some graph realizations to approximate the true objective and run submodular maximization algorithms on the sampled graphs. In contrast, our setting and algorithm do not assume that a graph is stored or the model of influence is explicitly known, only that we could simulate it for a subset. Therefore our approach could apply in more general influence maximization settings than the sampled realization approach.

Next, we describe the details about the three algorithms that we compare to:
(i) The fixed $\epsilon$ approximation (``\texttt{EPS-AP}'') algorithm. This is where we essentially run \alg, except instead of using the subroutine \samp to adaptively sample in order to reduce the number of samples, we simply sample down to an $\epsilon$-approximation of every marginal gain. This takes $N_1$ samples for every marginal gain computation, see definition of $N_1$ in Algorithm \ref{alg:samp}. The element $u$ is added to $S$ if and only if the empirical estimate $\widehat{\Delta f_{N_1}}(S,u)\geq w$; (ii)
The special case of the algorithm \singla of \citet{singla2016noisy} that yields about a $(1-1/e)$-approximate solution with high probability, ``\texttt{EXP-GREEDY}'', which is described in Section \ref{sec:relatedwork} and in the appendix. In the detailed description of \singla found in the appendix in the supplementary material, this is the case that $k'$ is set to be 1; (iii) The randomized version of the algorithm of \singla, ``\texttt{EXP-GREEDY-K}'', which yields about a $(1-1/e)$-approximation guarantee in expectation. Since \texttt{EXP-GREEDY-K} is a randomized algorithm, we average the results for \texttt{EXP-GREEDY-K} over $10$ trials. This is the case that $k'=\kappa$.

Then we provide some additional details for experiments on instances of data summarization. 
The parameter $\delta$ for all the experiments is set to be $0.2$, and the approximation precision parameter $\alpha$ is $0.2$ for both \alg and \texttt{EPS-AP}. 
The value of $\epsilon$ of the experiments for different $\kappa$ are $0.1$, $0.2$, $0.1$ and $0.1$ on corel\_60, delicious\_300, delicious, and corel respectively. The value of $\kappa$ for different $\epsilon$ are $10$, $80$, $200$ and $100$ on corel\_60, delicious\_300, delicious and corel respectively.

At last, we introduce the experimental setup for influence maximization. We run the four algorithms described above on the experiments for different values of $\kappa$ and $\epsilon$. The dataset used here is a sub-graph extracted from the EuAll dataset with $n=29$ \cite{leskovec2016snap}. The underlying weight of each edge is uniformly sampled from $[0,1]$ (``euall''). In our experiments, we simulate the influence maximization under the influence cascade model. We further use the reverse influence sampling (RIS) \cite{borgs2014maximizing} to enhance the computation efficiency of our algorithm. Here $R$ is the number of nodes in the graph and is thus $29$. The value of $\kappa$ for different $\epsilon$ is $8$, and the value of $\epsilon$ for different $\kappa$ is $0.15$. The parameters $\delta$ and $\alpha$ are set to be $0.2$ for both of the experiments. Since \texttt{EXP-GREEDY-K} is a randomized algorithm, the experimental results for \texttt{EXP-GREEDY-K} are averaged over $4$ trials for different $\epsilon$, and  $8$ trials for different $\kappa$.

\subsection{Addtional Experimental Results}
\label{results}
First, we present the result analysis of the experiments where we vary $\epsilon$. It can be seen from Figures \ref{fig:cover2500_300_eps_q}, \ref{fig:cover2500_300_eps_average-q}, \ref{fig:corel_60_eps_q} and \ref{fig:corel_60_eps_ave_q} that both the total samples and average samples of our algorithm \alg increase less compared with \texttt{EPS-AP} and \texttt{EXP-GREEDY} as $\epsilon$ decreases. 
This is not surprising, because the theoretical guarantee on the number of samples taken per marginal gain contribution in \texttt{EPS-AP} is $O(\frac{1}{\epsilon^2})$, 
which would increase rapidly when $\epsilon$ decreases. This also makes sense for \texttt{EXP-GREEDY}, since the theoretical guarantee on the number of queries of each iteration is $O(\frac{nR^2}{\epsilon^2}\log\big(\frac{R^2kn}{\delta\epsilon^2}\big))$ if the difference between elements marginal gains are very small.

Then we present the additional experimental results with respect to the function value $f$ on the instance of data summarization in the main paper. The results are in Figure \ref{fig:exp_results_of_f}. The experimental results of $f$ for different $\kappa$ are in Figure  \ref{fig:cover2500_300_k_f}, \ref{fig:cover_k_f},  \ref{fig:corel_60_k_f} and  \ref{fig:corel_k_f}. From the results, one can see that the $f$ values for different algorithms are very almost the same in most cases. However, when $\kappa$ increases and becomes large, the $f$ value of \texttt{EXP-GREEDY-K} is smaller than other algorithms, which is because when  $\kappa$ is large, it allows for more randomness in \texttt{EXP-GREEDY-K} and is  less accurate.

Next, we present the experimental results on the instance of influence maximization. The results are plotted in Figure \ref{fig:exp_results_of_influmax}. From the results, we can see that our proposed algorithm \alg outperforms the other three algorithms in terms of the total number of samples (see Figure \ref{fig:influ_eps_q}, \ref{fig:influ_k_q}). When $\kappa$ increases, the average number of samples decreases fast for \alg. This is because the marginal gain on this instance decreases rapidly when $\kappa$ increases while the threshold value decreases only by a factor of $1-\alpha$ at the end of each iteration, in many iterations the threshold value $w$ is much higher than the marginal gain and thus the gap function $\phi(S,s)$ is large. According to the results of sample complexity in Theorem \ref{mainthm}, the number of required samples decreases fast as $\kappa$ increases. This is also why the average number of samples of \alg is much smaller than \texttt{EXP-GREEDY} and \texttt{EXP-GREEDY-K} as is presented in Figure \ref{fig:influ_eps_ave_q} and Figure \ref{fig:influ_k_ave_q}.
\begin{figure*}[t!]
    \centering
    \hspace{-0.5em}
     \subfigure[delicious\_300 $f$]
{\label{fig:cover2500_300_eps_f}\includegraphics[width=0.24\textwidth]{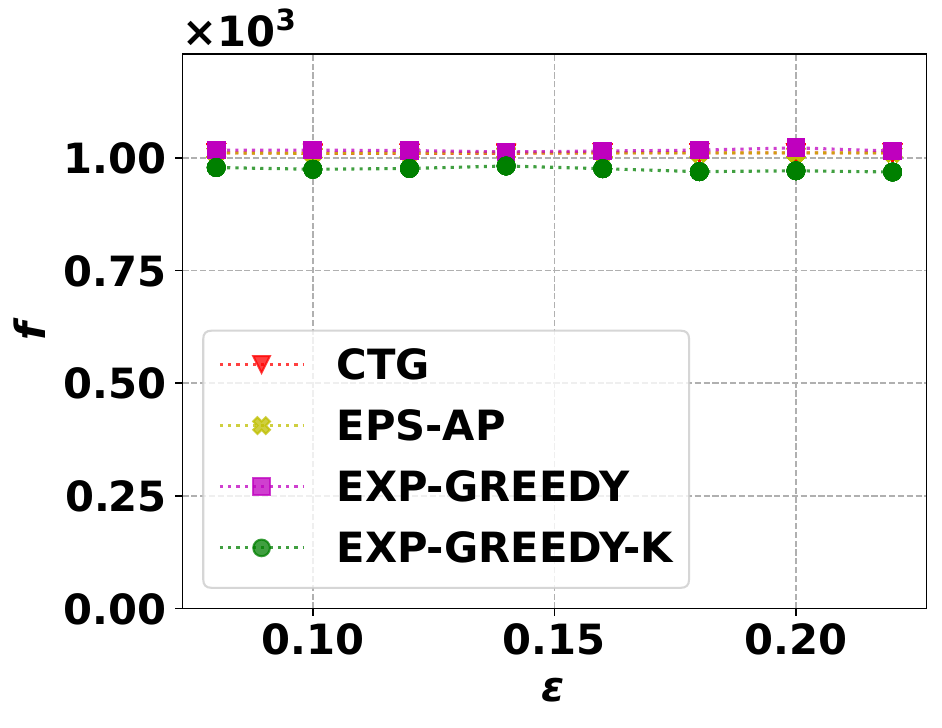}} 
\hspace{-0.5em}
     \subfigure[delicious\_300 $f$]
{\label{fig:cover2500_300_k_f}\includegraphics[width=0.24\textwidth]{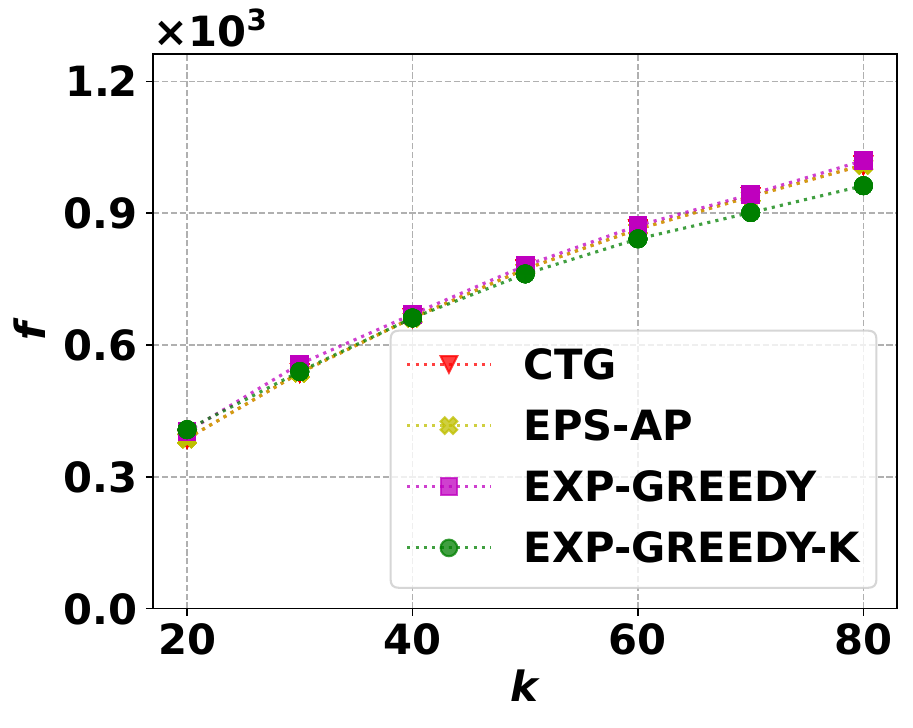}} 
    \hspace{-0.5em}
     \subfigure[corel\_60 $f$]
{\label{fig:corel_60_eps_f}\includegraphics[width=0.24\textwidth]{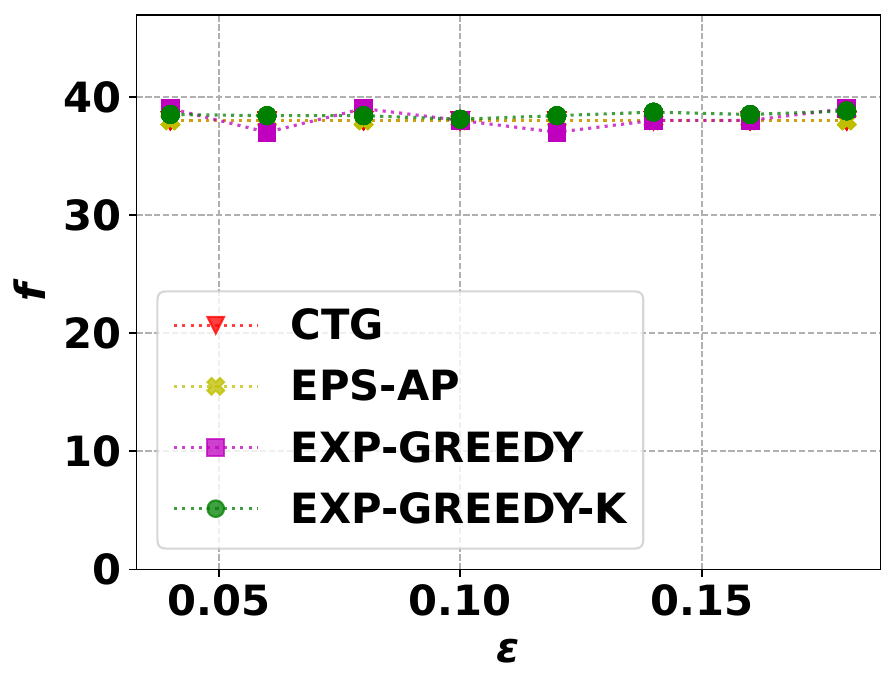}} 
\hspace{-0.5em}
     \subfigure[corel\_60 $f$]
{\label{fig:corel_60_k_f}\includegraphics[width=0.24\textwidth]{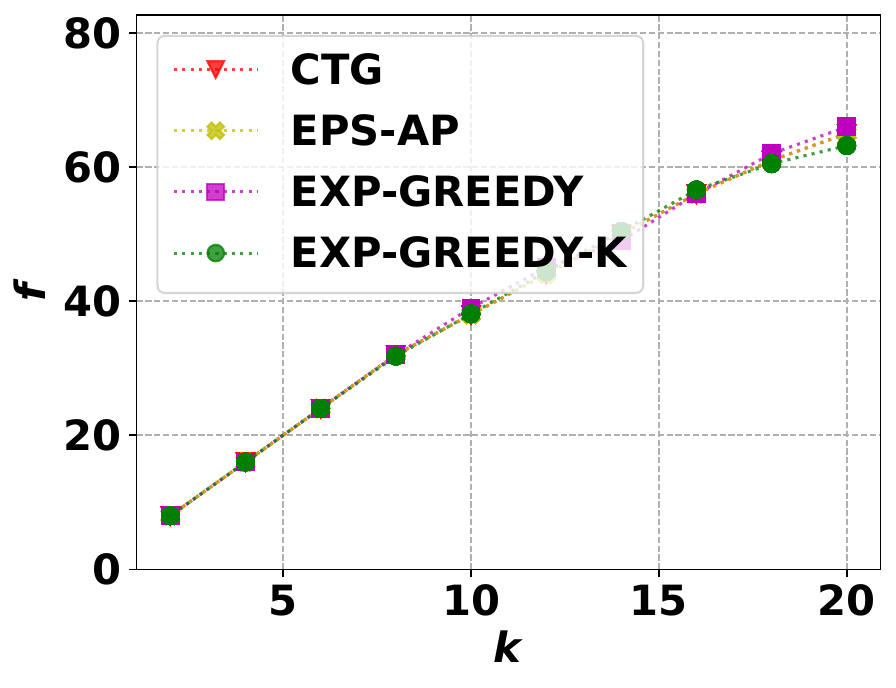}} 
\hspace{-0.5em}
     \subfigure[corel $f$]
{\label{fig:corel_eps_f}\includegraphics[width=0.24\textwidth]{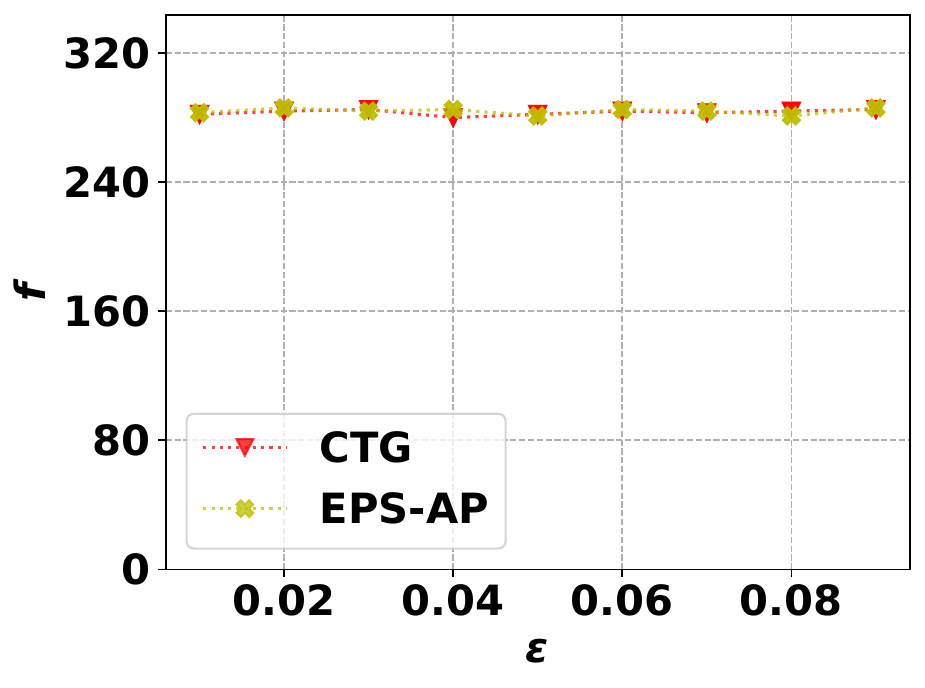}} 
\hspace{-0.5em}
     \subfigure[corel $f$]
{\label{fig:corel_k_f}\includegraphics[width=0.24\textwidth]{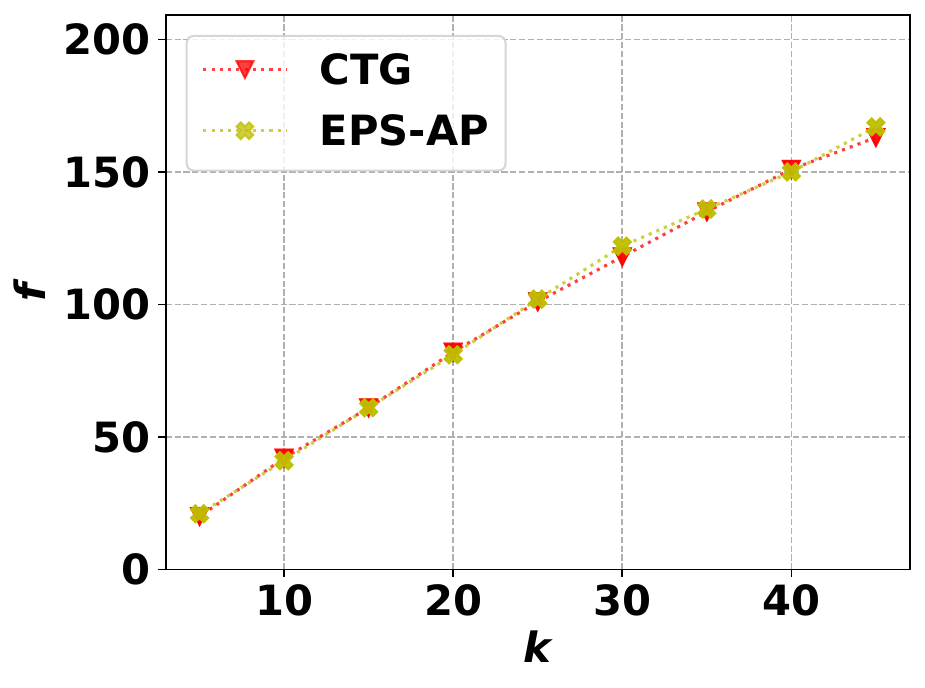}} 
\hspace{-0.5em}
     \subfigure[delicious $f$]
{\label{fig:cover_eps_f}\includegraphics[width=0.24\textwidth]{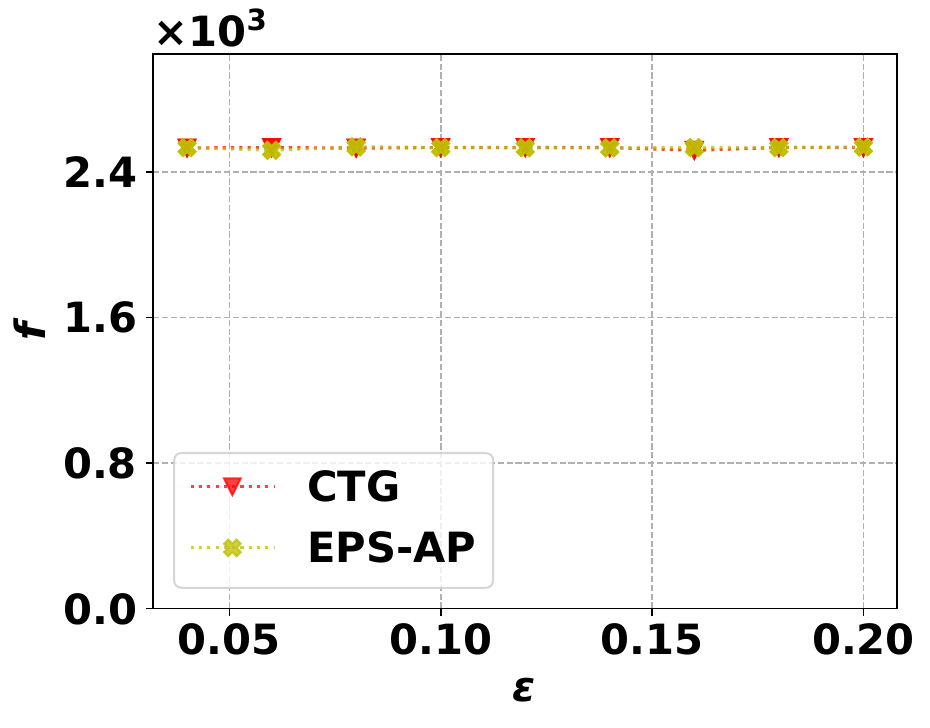}} 
\hspace{-0.5em}
     \subfigure[delicious $f$]
{\label{fig:cover_k_f}\includegraphics[width=0.24\textwidth]{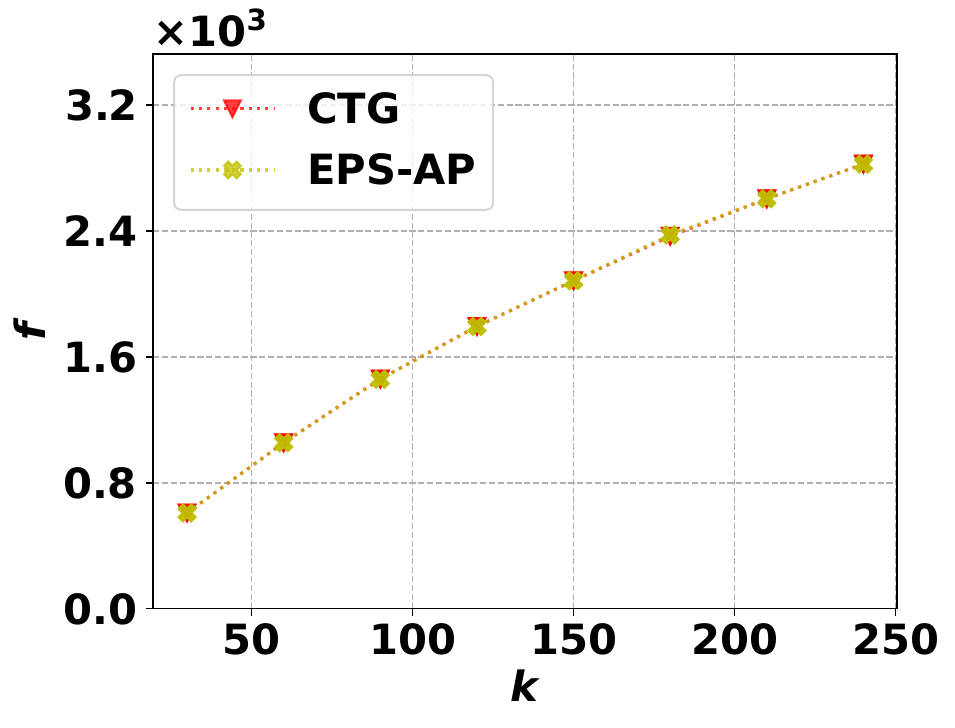}} 
\caption{The experimental results of $f$ of running different algorithms on instances of data summarization on the delicious URL dataset ("delicious", "delicious\_300") and Corel5k dataset ("corel", "corel\_60").}
\label{fig:exp_results_of_f}
\end{figure*}

\begin{figure*}[t!]
    \centering
    \hspace{-0.5em}
     \subfigure[euall samples]
{\label{fig:influ_eps_q}\includegraphics[width=0.24\textwidth]{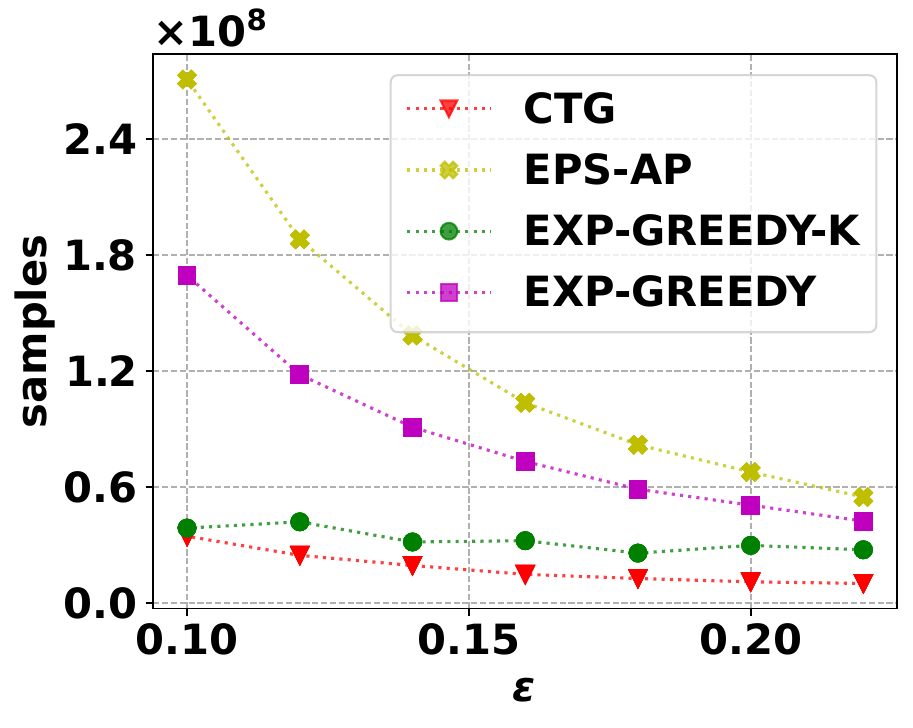}} 
\hspace{-0.5em}
     \subfigure[euall average samples]
{\label{fig:influ_eps_ave_q}\includegraphics[width=0.24\textwidth]{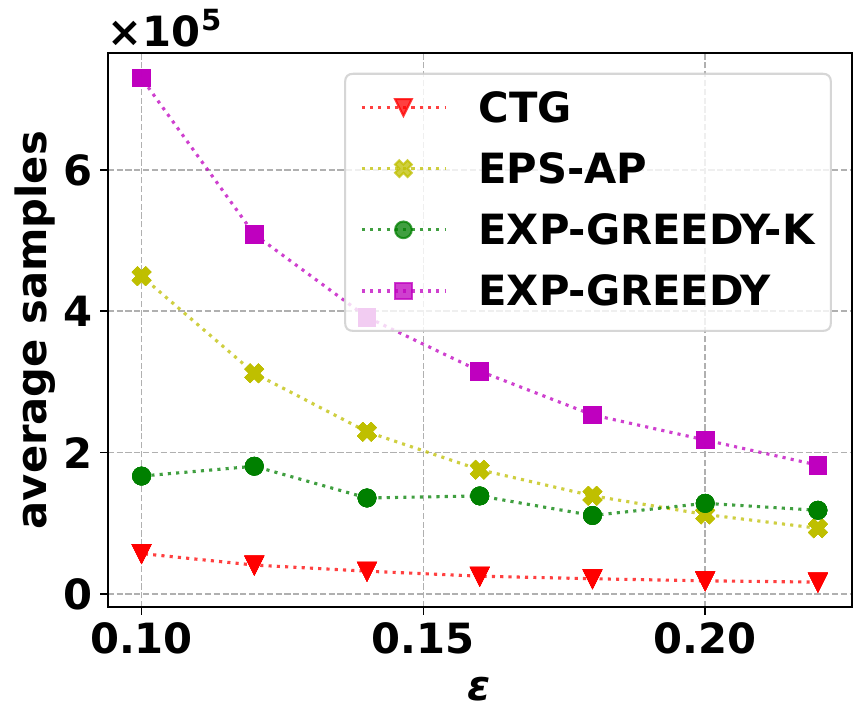}} 
    \hspace{-0.5em}
     \subfigure[euall $f$]
{\label{fig:influ_eps_f}\includegraphics[width=0.24\textwidth]{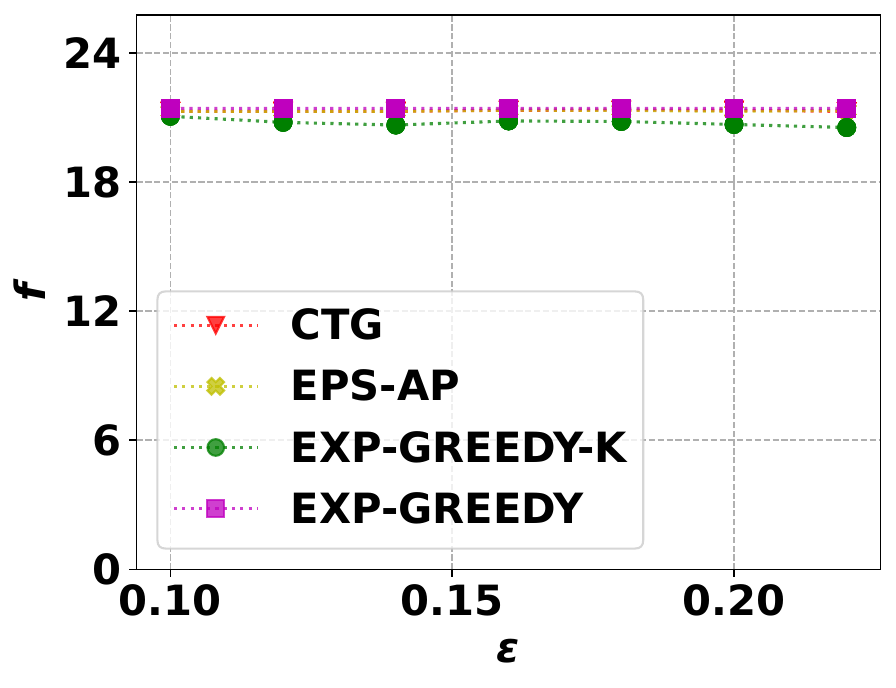}} 
\hspace{-0.5em}
     \subfigure[euall samples]
{\label{fig:influ_k_q}\includegraphics[width=0.24\textwidth]{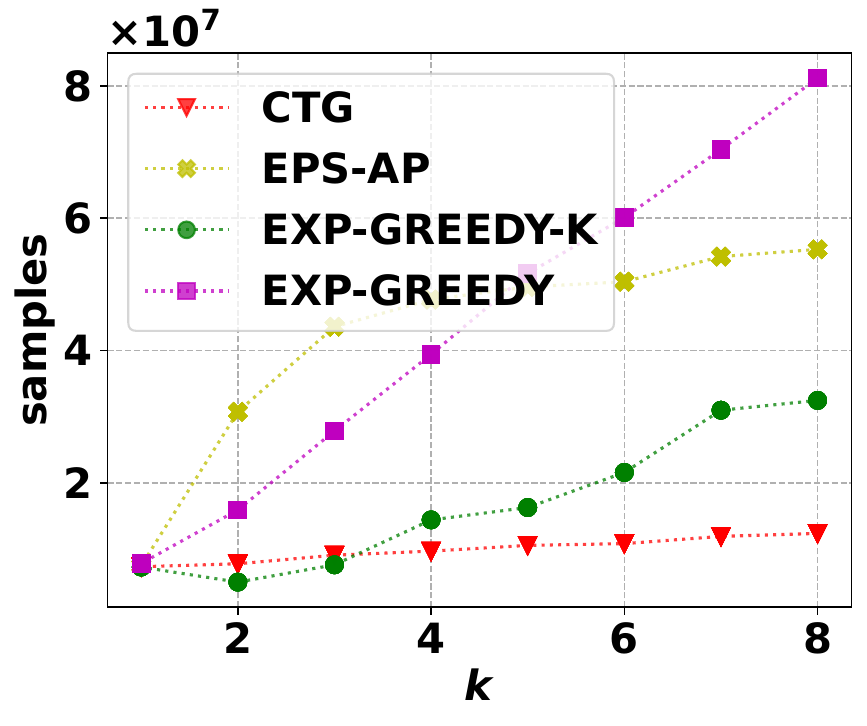}} 
\hspace{-0.5em}
     \subfigure[euall average samples]
{\label{fig:influ_k_ave_q}\includegraphics[width=0.24\textwidth]{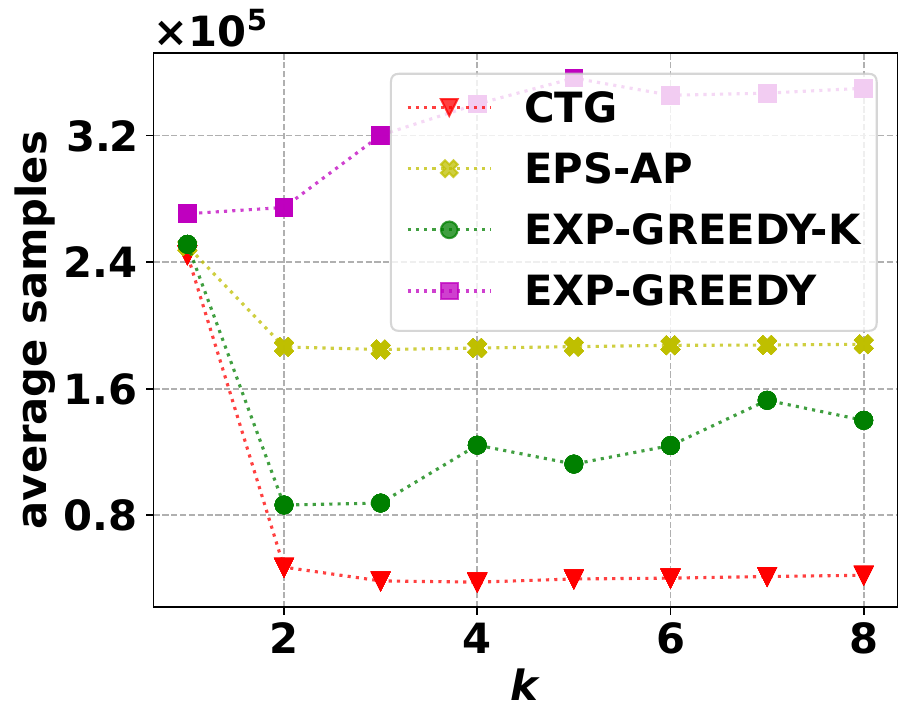}} 
\hspace{-0.5em}
     \subfigure[euall $f$]
{\label{fig:influ_k_f}\includegraphics[width=0.24\textwidth]{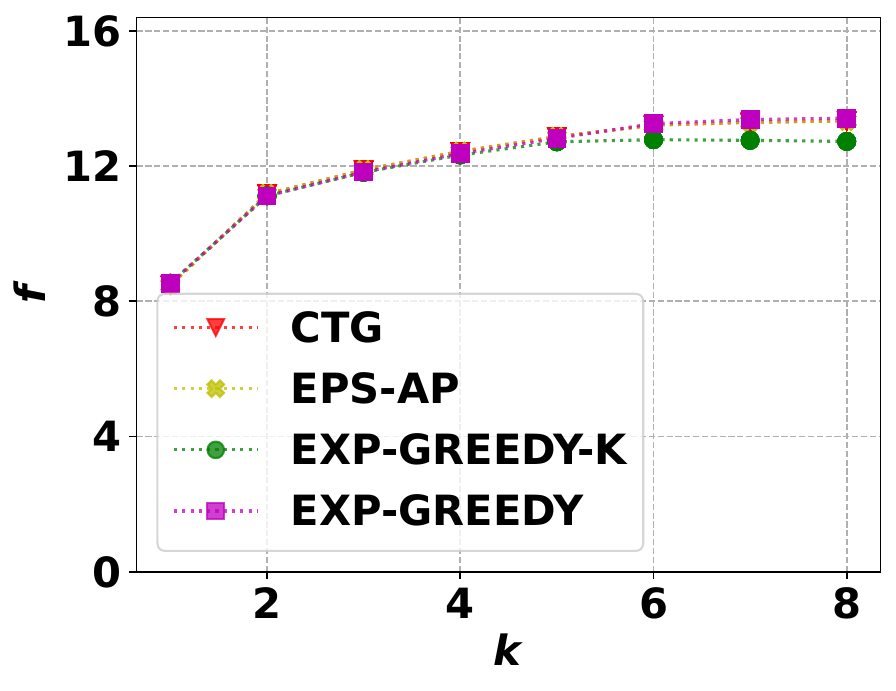}} 
\caption{The experimental results of running different algorithms on the instance of influence maximization on the EuAll dataset ("euall").}
\label{fig:exp_results_of_influmax}
\end{figure*}

\end{document}